\documentclass[letterpaper]{article} % DO NOT CHANGE THIS
\usepackage[draft]{aaai25}  % DO NOT CHANGE THIS
\usepackage{times}  % DO NOT CHANGE THIS
\usepackage{helvet}  % DO NOT CHANGE THIS
\usepackage{courier}  % DO NOT CHANGE THIS
\usepackage[hyphens]{url}  % DO NOT CHANGE THIS
\usepackage{graphicx} % DO NOT CHANGE THIS
\usepackage{natbib}  % DO NOT CHANGE THIS AND DO NOT ADD ANY OPTIONS TO IT
\usepackage{caption} % DO NOT CHANGE THIS AND DO NOT ADD ANY OPTIONS TO IT
\newif\ifarxiv
\arxivfalse

\usepackage[utf8]{inputenc} % allow utf-8 input

\usepackage{amssymb}
\usepackage{mathtools}
\usepackage{amsthm}
\usepackage{dsfont}
\usepackage{enumitem}
\usepackage{xspace}

\usepackage{url}            % simple URL typesetting
\usepackage{booktabs}       % professional-quality tables
\usepackage{amsfonts}       % blackboard math symbols
\usepackage[cmyk]{xcolor}         % colors
\usepackage{multirow}

\usepackage{subcaption}
\usepackage[algo2e,noend,ruled,linesnumbered]{algorithm2e} % algo2e avoids clash with "algorithm" package.

\usepackage[sort&compress,nameinlink,noabbrev,capitalize]{cleveref}

\theoremstyle{plain}
\newtheorem{theorem}{Theorem}[section]

\newtheorem{lemma}[theorem]{Lemma}
\newtheorem{rrule}[theorem]{Reduction Rule}

\theoremstyle{definition}
\newtheorem{definition}[theorem]{Definition}

\theoremstyle{remark}

\usepackage[disable]{todonotes}

\newcommand{\OS}{\texttt{Witty}}

\newcommand{\Thr}{\textsf{Thr}}

\newcommand{\cla}{\ensuremath{\textsf{cla}}}
\newcommand{\nISack}{\ensuremath{\textsf{nISack}}}
\newcommand{\sack}{\ensuremath{\textsf{sack}}}
\newcommand{\Oh}{\mathcal{O}}

\newcommand{\cut}{\ensuremath{\textsf{cut}}}
\newcommand{\Cuts}{\ensuremath{\textsf{Cuts}}}
\newcommand{\wit}{\ensuremath{\textsf{wit}}}
\newcommand{\leaf}{\ensuremath{\textsf{leaf}}}
\newcommand{\dirty}{\ensuremath{\textsf{Dirty}}}
\newcommand{\refine}{\ensuremath{\textsf{Ref}}}
\newcommand{\refineAll}{\ensuremath{\textsf{RefAll}}}

\newcommand{\imp}{\ensuremath{\textsf{imp}}}

\newcommand{\pairsplit}{\ensuremath{\textsf{pairsplit}}}
\newcommand{\pairs}{\ensuremath{\textsf{Pairs}}}
\newcommand{\tree}{\ensuremath{\textsf{Tree}}}
\newcommand{\ex}{\ensuremath{\textsf{ex}}}
\newcommand{\const}{\ensuremath{\textsf{Const}}}

\newcommand{\pDTSlong}{\textsc{Minimum-Size Decision Tree}}

\usepackage{scrextend}          %Labeling environment

\newcommand{\prob}[6]{%
  \begin{quote}
  	\begin{samepage}
    \begin{labeling}{#6}%
      \setlength\topsep{-.6ex} \setlength\itemsep{-.2ex}
    \item[#1]
    \item[\emph{#2}]#3
    \item[\emph{#4}]#5
    \end{labeling}%
	\end{samepage}
  \end{quote}%
}

\newcommand{\probdefopt}[3]{\prob{#1}{Instance:}{#2}{Task:}{#3}{as}}

\usepackage{etoolbox}
\newcommand{\appsymb}{$\bigstar$}
\newcommand{\appref}[1]{{\appsymb}}

\newcommand{\appendixsection}[1]{%
	\ifarxiv{}\else{}
	\gappto{\appendixProofText}{\section{Additional Material for \Cref{#1}}\label{app:#1}}
	\fi{}
}

\newcommand{\toappendix}[1]{%
	\ifarxiv{}#1\else{}
	\gappto{\appendixProofText}
	{{
			#1
	}}
	\fi{}
}

\newcommand{\appendixproof}[2]{%
	\ifarxiv{}#2\else{}\gappto{\appendixProofText}
	{
		\subsection{Proof of \Cref{#1}}\label{proof:#1}
		#2
	}
	\fi{}
}

\title{Witty: An Efficient Solver for Computing Minimum-Size Decision Trees}
\author{
    Luca Pascal Staus\textsuperscript{\rm 1},
    Christian Komusiewicz\textsuperscript{\rm 1},
    Frank Sommer\textsuperscript{\rm 2},
    Manuel Sorge\textsuperscript{\rm 2}
}
\affiliations {
    \textsuperscript{\rm 1}Institute of Computer Science, University of Jena, Germany\\
    \textsuperscript{\rm 2}Institute of Logic and Computation, TU Wien, Austria\\
    luca.staus@uni-jena.de, c.komusiewicz@uni-jena.de, fsommer@ac.tuwien.ac.at, manuel.sorge@ac.tuwien.ac.at
}

\begin{document}

\maketitle

\begin{abstract}
  %\looseness=-1
  Decision trees are a classic model for summarizing and classifying data. To enhance interpretability and generalization properties, it has been proposed to favor small decision trees.  
  Accordingly, in the minimum-size decision tree training problem (\textsc{MSDT}), the input is a set of training examples in $\mathbb{R}^d$ with class labels and we aim to find a decision tree that classifies all training examples correctly and has a minimum number of nodes.
  \textsc{MSDT} is NP-hard and therefore presumably not solvable in polynomial time.
  Nevertheless, Komusiewicz et al. [ICML~'23] developed a promising algorithmic paradigm called \emph{witness trees} which solves \textsc{MSDT} efficiently if the solution tree is small.
  In this work, we test this paradigm empirically.
  We provide an implementation, augment it with extensive heuristic improvements, and scrutinize it on standard benchmark instances.
  The augmentations achieve a mean 324-fold (median 84-fold) speedup over the naive implementation. Compared to the state of the art they achieve a mean 32-fold (median 7-fold) speedup over the dynamic programming based MurTree solver [Demirović et al., J.~Mach.~Learn.~Res.~'22] and a mean 61-fold (median 25-fold) speedup over SAT-based implementations [Janota and Morgado, SAT~'20].
  As a theoretical result we obtain an improved worst-case running-time bound for \textsc{MSDT}.
\end{abstract}

\begin{links}
	\link{Code, Data and Experimental Results}{https://doi.org/10.5281/zenodo.11235017}
\end{links}

\section{Introduction}
\label{sec-intro}
\appendixsection{sec-intro}
\looseness=-1
% Decision trees are tools for summarizing and classifying data.
Traditional decision trees recursively partition the feature space with axis-parallel binary cuts and assign a class to each part of the partition.
When learning decision trees, we are given a training data set that consists of examples $E \subseteq \mathbb{R}^d$ labeled by classes and we want to find a decision tree that classifies the training data set (see \Cref{sec:Preliminaries} for the formal definitions).
In addition, we want to optimize certain criteria, chiefly among them we want to minimize the number of cuts in the tree (or, equivalently, the size of the tree).\footnote{Another popular choice is to minimize the depth of the tree. We opted for the size criterion instead for simplicity. Our techniques are broadly applicable to the depth criterion, see the discussion in \cref{sec:conclusion}.}
This is because small trees are more easily interpretable~\citep{molnar_interpretable_2020,moshkovitz_explainable_2020} and they are thought to be more accurate on unknown data~\citep{fayyad_what_1990,Hu0HH20}.
The resulting optimization problems are NP-hard, however~\citep{hyafil_constructing_1976,OrdyniakS21}.
Therefore, researchers traditionally use heuristics such as CART~\citep{DBLP:books/wa/BreimanFOS84}.
Recent advancements in hardware and algorithms have made it feasible to compute provably optimal trees for training data sets with up to several hundreds of examples % ~\citep{DBLP:conf/cp/BessiereHO09,NijssenF10,VerwerZ17,bertsimas_optimal_2017,DBLP:conf/ijcai/NarodytskaIPM18,DBLP:conf/aaai/VerwerZ19,HuRS19,AGV20,ZhuMPNK20,FiratCGHZ20,DBLP:conf/sat/JanotaM20,AglinNS20,avellaneda_efficient_2020,Hu0HH20,LinZHRS20,VerhaegheNPQS20,GunlukKLMS21,DBLP:journals/jmlr/DemirovicLHCBLR22,BoutilierMZ22,McTavishZAKCRS22,LindenWD23,SCM23,SchidlerS24,CRB24}
(see the surveys by \citet{carrizosa_mathematical_2021,DBLP:journals/air/CostaP23} and a full list of 24 implementations in the appendix).
The implementations still suffer from poor scalability, however~\citep{DBLP:journals/air/CostaP23}.
In this paper we contribute towards overcoming this limitation.

%\looseness=-1
To that end, as has been done before~\citep{DBLP:conf/sat/JanotaM20,DBLP:conf/ijcai/NarodytskaIPM18,DBLP:journals/jmlr/DemirovicLHCBLR22}, we focus on the training problem in which we want to find the smallest size perfect decision tree for the input training examples from two classes.
\emph{Perfect} means that all examples are classified correctly.
This choice is to create a baseline method and we aim to later extend our approach to weaker accuracy guarantees and more classes.
We build on recent theoretical algorithmic research that investigated what properties of the training data make computing smallest decision trees hard or tractable~\citep{OrdyniakS21,EibenOrdyniakPaesaniSzeider23,KobourovLMPRSSW23,KKSS23}.
Our starting point is the algorithmic paradigm of \emph{witness trees}, introduced by \citet{KKSS23}.
This paradigm improved on previous approaches~\citep{OrdyniakS21,EibenOrdyniakPaesaniSzeider23} with a better running time guarantee of $\Oh((6\delta D s)^s \cdot s n)$ in terms of the number $n$ of training examples and the small parameters largest domain size $D$ ($D=2$ is common), largest number $\delta\le d$ of features in which two training examples differ,\footnote{See \citet{OrdyniakS21} and \Cref{tab:overview_datasets_reduced} in the appendix for evidence that $\delta$ is a reasonably small parameter (much smaller than the number~$d$ of dimensions) in several datasets.}
and the number $s$ of cuts in the solution tree.
Importantly for us, the paradigm simplified previous theoretical approaches to such an extent as to make it promising as a foundation of a competitive solver.
Additionally, the approach generalizes to other learning models and even ensembles~\citep{KKSS23,OPRS24}, making it particularly interesting to test in practice.

The main idea of the witness-tree paradigm is to start with a trivial decision tree consisting of a single leaf labeled by an arbitrary class. Then we pick a training example which is classified incorrectly, we call it \emph{dirty}, and we recursively try in a branching over all possibilities to refine the current tree with a new cut and a new leaf in which the dirty example is classified correctly.
We label the new leaf with the dirty example as \emph{witness} and mandate that all future refinements maintain all witnesses at their assigned leafs.
This reduces the search space that we need to explore to find the desired decision tree.
See \Cref{sec:base_algo} for a more detailed description.
In this work, we provide an implementation of the above paradigm and we extensively engineer it with improvements that give heuristic speedups while maintaining the optimality of the computed decision tree.

%\looseness=-1
Our main result is that these improvements substantially speed up the algorithm to the point where the resulting solver, called \OS{}, performs substantially better than the state of the art (see \Cref{sec:evaluation}).
To the best of our knowledge, the state-of-the-art solvers for computing minimum-size perfect decision trees are SAT-based solvers by \citet{DBLP:conf/ijcai/NarodytskaIPM18} and \citet{DBLP:conf/sat/JanotaM20} and the dynamic-programming algorithm \texttt{MurTree} by \citet{DBLP:journals/jmlr/DemirovicLHCBLR22}.
All other implementations for computing variants of optimal decision trees solve different optimization problems and are not suitable for the task discussed in the present paper, see the appendix for details.

We tested the solvers on standard benchmark data sets from the Penn Machine Learning Benchmarks~\citep{romano2021pmlb}.
Out of the three state-of-the-art solvers \texttt{MurTree} performed the best, solving~$371$ out of the~$700$ instances. In comparison, \OS{} was able to solve~$388$ instances and achieves a mean~$32$-fold (median~$7$-fold) speedup over \texttt{MurTree}. \OS{} performs especially well compared to \texttt{MurTree} on instances that have dimensions with a large number of different values. This is due to the fact that \OS{} does not require binary dimensions and one of the above-mentioned improvements specifically exploits large domains.
We should also note the following. For the calculation of all speedup factors in this paper we ignored all instances that were solved in less than one second by both algorithms as these instances are not a good indicator for the scaling behavior of the running times.

%In comparison, we achieve a mean 61-fold (median 25-fold) speedup on standard benchmark data sets (part of the Penn Machine Learning Benchmarks~\citep{romano2021pmlb}).
%Indeed, \OS{} is faster than the state of the art on all but 8 of the 389 instances solved by either algorithm within the 60\,min time limit.
%Furthermore, \OS{} is efficient enough such that it can find optimal solution trees of size up to 16 whereas the state of the art could find optimal solutions of size up to 12.% \todo{ck: First, cuts have not been introduced, second this sounds like we would like to maximize no. of cuts.\\ms: Cuts are mentioned before.}

%\looseness=-1
To achieve these results, we employ the following techniques.
First, we apply data reduction rules that simplify the instances (\cref{sec:base_algo}).
Second, we improve the branching rule which introduces the next cut into the partial decision tree in several ways (\cref{sec:base_algo}):
We carefully select which dirty examples to use for branching, specify a suitable sequence of the cuts to try, and observe that some cuts can be omitted from branching.
Third, we introduce lower-bounding strategies that determine a minimum number of cuts that still need to be added to the tree, to further shrink the search space (\cref{sec:lower_bounds}).
Fourth, we use symmetry-breaking techniques which we call subset constraints that leverage information from unsuccessfully returning branches which essentially fixes some examples into subtrees of the solution decision tree (\cref{sec:subset_constraints}).
This also yields an improved running-time bound of $\Oh((\delta D \log s)^s \cdot s n)$ (\cref{thm-dts-improved-algo}).
Finally, during the search we select certain subsets of examples for which we directly compute lower bounds, cache them, and exploit them later in the search (\cref{sec:subset_caching}).
We rigorously analyze all of the above techniques and prove them to preserve optimality of the computed trees.

%\looseness=-1
Apart from improved scalability, other key advantages of our implementation include that, different from the state-of-the-art algorithms, it is not necessary to binarize the features, and that the general paradigm is very flexible so that it can be adapted for reoptimizing a given decision tree, other optimization goals such as minimizing the number of misclassified examples, and for other concepts related to decision trees such as decision lists.
Similarly, the above heuristic improvements are flexible and they apply to computing minimum-depth trees, for example.
In summary, we provide an implementation and extensive heuristic tuning of the witness-tree paradigm for computing minimum-size decision trees which is flexible and at the same time substantially outperforms the state of the art.

\toappendix{

  \paragraph{Implementations for computing optimal decision trees and state of the art for computing minimum-size perfect trees.}
  We are aware of the following 26 implementations that solve exactly some form of optimization problem related to computing small decision trees:
  \citet{DBLP:conf/cp/BessiereHO09,NijssenF10,VerwerZ17,bertsimas_optimal_2017,DBLP:conf/ijcai/NarodytskaIPM18,DBLP:conf/aaai/VerwerZ19,HuRS19,AGV20,ZhuMPNK20,FiratCGHZ20,DBLP:conf/sat/JanotaM20,AglinNS20,avellaneda_efficient_2020,Hu0HH20,LinZHRS20,VerhaegheNPQS20,GunlukKLMS21,DBLP:journals/jmlr/DemirovicLHCBLR22,BoutilierMZ22,McTavishZAKCRS22,LindenWD23,SCM23,DBLP:conf/icml/DemirovicHJ23,DBLP:journals/air/AlosAT23,SchidlerS24,CRB24}.
  
  Out of these, \citet{DBLP:conf/sat/JanotaM20,DBLP:conf/cp/BessiereHO09,NijssenF10,DBLP:conf/ijcai/NarodytskaIPM18,HuRS19,AGV20,avellaneda_efficient_2020,LinZHRS20,GunlukKLMS21,DBLP:journals/jmlr/DemirovicLHCBLR22,McTavishZAKCRS22,LindenWD23,DBLP:journals/air/AlosAT23,SchidlerS24,CRB24} would theoretically be suitable to minimize the size, that is, the number of inner nodes, of perfect trees, see the more detailed discussion below.
  
  All of the remaining implementations are not suitable for this optimization goal:
  Most of the remaining implementations minimize some form of classification error under a given depth constraint on the solution tree \cite{VerwerZ17,DBLP:conf/aaai/VerwerZ19,FiratCGHZ20,Hu0HH20,VerhaegheNPQS20,McTavishZAKCRS22,SCM23,DBLP:conf/icml/DemirovicHJ23}.
  It would thus be possible to find perfect trees using these implementations, but it is not possible to minimize the number of inner nodes.
  Note that, minimizing the size of the tree is a more difficult optimization problem than minimizing the depth because many more different topologies are possible and have to be tested in the search for the optimum.
  The remaining three implementations either do not consider axis-aligned cuts \citep{bertsimas_optimal_2017,BoutilierMZ22} as we do, or instead optimize some clustering property \citep{ZhuMPNK20}.

  Of the implementations that can minimize the size for perfect trees, we further had to exclude the following: \citet{DBLP:conf/cp/BessiereHO09} was superseded by the more efficient implementation by \citet{DBLP:conf/ijcai/NarodytskaIPM18}.
  Similarly, \citet{NijssenF10} was superseded by \citet{AglinNS20} which was in turn superseded by the more efficient implementation by \citet{DBLP:journals/jmlr/DemirovicLHCBLR22}.
  Moreover, \citet{HuRS19} was superseded by \citet{LinZHRS20} (see also \citet{DBLP:journals/jmlr/DemirovicLHCBLR22}) which was in turn superseded by the more efficient implementation by \citet{McTavishZAKCRS22}.
  These three implementations minimize a sum of a classification error and a given real parameter $\lambda$ times the number of leaves of the solution tree, perhaps with an additional depth constraint.
  By setting the depth constraint large enough and the parameter $\lambda$ small enough these implementations hence find minimum-size perfect trees.
  However, we tested the most efficient implementation with the corresponding parameter values and, apparently due to not being optimized for such parameters, it was clear that it is by far not competitive for finding minimum-size perfect trees.
  For similar reasons, we had to exclude \citet{CRB24} because prelimary experiments showed that it is also clearly not competitive for the parameter settings that lead to finding minimum-size perfect trees.
  
  In the implementation by \citet{GunlukKLMS21} the topology of the tree is given as an input.
  One could thus find the minimum size by running their algorithm for all possible topologies.
  However, it is clear that this cannot lead to a competitive solver, hence we excluded it.
  The implementation by \citet{LindenWD23} solves a much more general problem that can minimize any so-called separable optimization function.
  However, apparently due to the generality, it was found to be less efficient than \citet{DBLP:journals/jmlr/DemirovicLHCBLR22}'s, see \citet{SchidlerS24}.
  In \citet{DBLP:journals/air/AlosAT23} the authors mention that the performance of their algorithm is similar to \citet{DBLP:conf/sat/JanotaM20}.
  The remaining implementations we had to exclude because, while theoretically their approach could be used to minimize the size of the trees, no implementation of the size-minimization variant was available: \citet{avellaneda_efficient_2020} describes an algorithm that solves the decision problem of whether there exists a perfect decision tree that satisfies a given depth constraint and size constraint.
  By setting the depth large enough and incrementally increasing the size constraint one could thus find a minimum-size perfect tree.
  However, the implementation by \citet{avellaneda_efficient_2020} does not allow for such an approach.\footnote{Instead, either the implementation minimizes the depth or it finds the minimum-size perfect tree under the additional constraint that the depth is minmal.
  The minimum-size tree might not have minimal depth and thus the implementation does not find a minimum-size perfect tree in all cases.}
  Consequently, we had to exclude it as well.
  Similarly, \citet{AGV20}'s approach could be used to minimize the size of the trees but the available implementation does not allow for such a setting.
% \foonote{For several instances, the program crashed and it ran out of memory for a simple instance with 21 training examples.}
  While \citet{SchidlerS24} mention that their approach could be adapted to minimizing the size of the trees instead of the depth, we are not aware of an implementation of the size-minimization variant.

  Overall, the state of the art for minimizing the size of perfect decision trees thus comprises \citet{DBLP:conf/sat/JanotaM20,DBLP:conf/ijcai/NarodytskaIPM18,DBLP:journals/jmlr/DemirovicLHCBLR22}.

}

\section{Preliminaries}
\label{sec:Preliminaries}

For $n\in\mathds{N}$, we denote~$[n] \coloneqq \{ 1,2,\ldots ,n\}$. For a vector~$x\in \mathds{R}^d$, we denote by~$x[i]$ the~$i$th entry in~$x$.
Let~$\Sigma$ be a set of \emph{class symbols}. 
We consider binary classification, so we assume that~$\Sigma = \{ \text{blue}, \text{red} \}$. 
A \emph{data set} with classes~$\Sigma$ is a tuple $(E, \lambda)$ of a set of \emph{examples} $E \subseteq \mathds{R}^d$ and their class labeling $\lambda \colon E \to \Sigma$.
Note that this formulation captures ordered features because such features can be mapped into~$\mathds{R}$.
We assume that~$(E,\lambda)$ does \emph{not} contain two examples with identical coordinates but different class labels.
For a fixed data set~$(E,\lambda)$, we let~$n\coloneqq |E|$  denote the \emph{number of examples} and~$d$ the \emph{dimension} of the data set.
%The maximum number of dimensions in which two differently labeled examples differ is denoted by~$\delta$.
%The maximum number of unique values in a dimension is denoted by~$D$.

For each dimension~$i\in [d]$, we let~$\Thr(i)$ be a smallest set of \emph{thresholds} such that for each pair of examples~$e_1, e_2\in E$ with~$e_1[i]<e_2[i]$ there is a threshold~$t\in \Thr(i)$ with~$e_1[i]\leq t < e_2[i]$. 
Note that such a set~$\Thr(i)$ can be computed in~$\Oh(n \log n)$~time and that $|\Thr(i)| \leq D$.
A \emph{cut} is a pair~$(i,t)$ where~$i\in [d]$ is a dimension and~$t\in \Thr(i)$ is a threshold in dimension~$d$. 
The set of all cuts is denoted by~$\Cuts(E)$. 
The \emph{left side} of a cut with respect to~$E'\subseteq E$ is~$E'[\leq (i,t)] \coloneqq \{ e\in E' \mid e[i] \leq t \}$, and the \emph{right side} of a cut with respect to~$E'$ is~$E'[> (i,t)] \coloneqq \{ e\in E' \mid e[i] > t \}$.

A \emph{decision tree} is a tuple~$\mathcal{D} = (T,\cut,\cla)$ where~$T$ is an ordered rooted binary tree with vertex set~$V(\mathcal{D})$, that is, each inner vertex has a well-defined left and right child. 
Furthermore, $\cut: V(\mathcal{D})\to \Cuts(E)$ maps every inner vertex to a cut and~$\cla: V(\mathcal{D})\rightarrow \Sigma$ labels each leaf.
The \emph{size} of $\mathcal{D}$ is the number of inner vertices.
For each vertex~$v\in V(\mathcal{D})$ we define a set~$E[\mathcal{D},v]\subseteq E$ of examples that are \emph{assigned} to~$v$.  
If~$v$ is the root of~$\mathcal{D}$, then~$E[\mathcal{D},v] \coloneqq E$. 
For other vertices, the assigned values are defined via the cuts at inner vertices. More precisely, for a parent vertex~$p$ with left child~$u$ and right child~$v$ we set~$E[\mathcal{D},u]\coloneqq E[\mathcal{D},p][\leq \cut(p)]$  and~$E[\mathcal{D},v]\coloneqq E[\mathcal{D},p][> \cut(p)]$. \todo[inline]{C: Still not a fan of the double [][] notation. Thoughts?}
If~$\mathcal{D}$ is clear from the context, we just write~$E[v]$. 
By definition, each example~$e\in E$ is \emph{assigned} to exactly one leaf~$\ell$ of~$\mathcal{D}$. 
We say that~$\ell$ is \emph{the leaf of~$e$} in~$\mathcal{D}$ and denote~$\ell$ by~$\leaf(\mathcal{D},e)$ or just~$\leaf(e)$ if~$\mathcal{D}$ is clear. 
An example $e\in E$ is \emph{dirty} in~$T$ if we have $\lambda(e)\ne \cla(\ell)$ with~$\ell$ being the leaf of $e$. 
The set of all dirty examples in $T$ is $\dirty(T)$. 
A decision tree \emph{classifies}~$(E,\lambda)$ if the class of every example~$e\in E$ matches the class of its leaf, that is, we have~$\lambda(e) = \cla(\leaf(e))$ for all~$e\in E$. 
In this case we call~$\mathcal{D}$ \emph{perfect}.

We aim to solve the following computational problem.

\probdefopt{\textsc{Minimum-Size Decision Tree (MSDT)}}
{A data set~$(E,\lambda)$.}
{Find a smallest decision tree that classifies~$(E,\lambda)$.}

We let \textsc{Bounded-Size Decision Tree (BSDT)} denote the variant where we are also given a number~$s$ and need to find a decision tree of size at most~$s$ or decide that no such tree exists. We can solve \textsc{MSDT} by solving multiple instances of \textsc{BSDT} as follows:
Start with~$s=1$. Increase~$s$ one by one and solve \textsc{BSDT} for each~$s$ until a tree is found.
This tree must be a solution for \textsc{MSDT}.
% \probdef{\textsc{Bounded-Size Decision Tree (BSDT)}}
% {A data set~$(E,\lambda)$ and~$s\in\mathds{N}$.}
% {Is there a decision tree of size at most~$s$ that classifies~$(E,\lambda)$?}

% A \emph{parameterized problem}~$L$ is a subset of~$\Sigma^*\times\mathds{N}$, where the first component is the input and the second is the parameter.
% We say that $L$ is \emph{fixed-parameter tractable}~(FPT) if each instance~$(x,k)$ can be solved in $f(k) \cdot \poly(|x|)$~time, where~$f$ is computable function only depending on~$k$.
% For more details about parameterized complexity, we refer to the standard monographs~\cite{flum_parameterized_2006,Nie06,downey_fundamentals_2013,CyFoKoLoMaPiPiSa2015}.

Proofs marked with~$(\bigstar)$ are deferred to the appendix.

\section{Base version -- the witness-tree algorithm}
\label{sec:base_algo}
\appendixsection{sec:base_algo}

The base version of our algorithm is the witness-tree algorithm of \citet{KKSS23}.
In the following, we briefly explain the algorithm. The presented algorithm solves the \textsc{BSDT} problem where we have a fixed size threshold~$s$. 
We then use this algorithm to solve \textsc{MSDT} as described above.

\paragraph{Witness trees and one-step refinements.}
Witness trees are decision trees augmented with a labeling of the leafs by examples. 
Formally, a \emph{witness tree} is a tuple~$W = (T,\cut,\cla,\wit)$~where $(T,\cut,\cla)$ is a decision tree and~$\wit: V(W)\rightarrow E$ is a map that maps each leaf~$\ell$ to an example~$e\in E[\ell]$ of the same class as~$\ell$. 
These examples are called \emph{witnesses} of~$W$. 
% All definitions of \Cref{sec:Preliminaries} also apply for witness trees.
%
During branching, $W$ is extended by adding a new inner vertex and a new leaf.
This is done via one-step refinements.
Formally, a \emph{one-step refinement} of~$W$ is a tuple~$(v,i,t,e)$ where~$v\in V(W)$ is some vertex in~$W$,~$(i,t)$ is a cut in~$\Cuts(E)$, and~$e\in E[W,v]$ is an example.

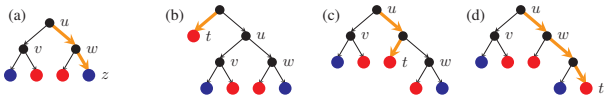
\begin{figure}
\centering
\scalebox{0.6}{
\begin{tikzpicture}[scale=0.29]
\newcommand{\sizeL}{8}
\newcommand{\sizeS}{6}

\node[label=right:{(a)}] (1) at (-4, 6.5){};
\draw [-stealth, orange, line width=2pt](0,6) -- (2,4.3);
\draw [-stealth, orange, line width=2pt](2,4) -- (3,2.3);
\node[fill,circle,inner sep=0pt,minimum size=\sizeS pt,label=right:{$u$}] (1) at (0, 6){};
\draw [-stealth](0,6) -- (-2,4.3);

\node[fill,circle,inner sep=0pt,minimum size=\sizeS pt,label=right:{$v$}] (1) at (-2, 4){};
\draw [-stealth](-2,4) -- (-3,2.3);
\draw [-stealth](-2,4) -- (-1,2.3);
\node[fill,circle,inner sep=0pt,minimum size=\sizeS pt,label=right:{$w$}] (1) at (2, 4){};
\draw [-stealth](2,4) -- (1,2.3);

\node[fill=blue,circle,inner sep=0pt,minimum size=\sizeL pt] (1) at (-3, 2){};
\node[fill=red,circle,inner sep=0pt,minimum size=\sizeL pt] (1) at (-1, 2){};

\node[fill=red,circle,inner sep=0pt,minimum size=\sizeL pt] (1) at (1, 2){};
\node[fill=blue,circle,inner sep=0pt,minimum size=\sizeL pt,label=right:{$z$}] (1) at (3, 2){};

%\node[label=right:{\large current tree $T$ with}] (1) at (-5.0, 0){};
%\node[label=right:{\large misclassified \textcolor{red}{red}}] (1) at (-4.2, -1.5){};
%\node[label=right:{\large example $e$ in leaf $z$}] (1) at (-4.8, -3){};

%\node[label=right:{\Huge $\rightarrow$}] (1) at (4.0, 5){};
%\node[label=right:{\large possible}] (1) at (5, 4){};
%\node[label=right:{\large refinements}] (1) at (4.5, 3){};

\node[label=right:{(b)}] (1) at (8, 6.5){};
\draw [-stealth, orange, line width=2pt](13,7) -- (11,5.3);
\node[fill,circle,inner sep=0pt,minimum size=\sizeS pt] (1) at (13, 7){};

\draw [-stealth](13,7) -- (15,5.3);
\node[fill=red,circle,inner sep=0pt,minimum size=\sizeL pt,label=right:{$t$}] (1) at (11, 5){};

\node[fill,circle,inner sep=0pt,minimum size=\sizeS pt,label=right:{$u$}] (1) at (15, 5){};
\draw [-stealth](15,5) -- (13,3.3);
\draw [-stealth](15,5) -- (17,3.3);

\node[fill,circle,inner sep=0pt,minimum size=\sizeS pt,label=right:{$v$}] (1) at (13.0, 3){};
\draw [-stealth](13,3) -- (12,1.3);
\draw [-stealth](13,3) -- (14,1.3);
\node[fill,circle,inner sep=0pt,minimum size=\sizeS pt,label=right:{$w$}] (1) at (17.0, 3){};
\draw [-stealth](17,3) -- (16,1.3);
\draw [-stealth](17,3) -- (18,1.3);

\node[fill=blue,circle,inner sep=0pt,minimum size=\sizeL pt] (1) at (12, 1){};
\node[fill=red,circle,inner sep=0pt,minimum size=\sizeL pt] (1) at (14, 1){};

\node[fill=red,circle,inner sep=0pt,minimum size=\sizeL pt] (1) at (16, 1){};
\node[fill=blue,circle,inner sep=0pt,minimum size=\sizeL pt] (1) at (18, 1){};

%\node[label=right:{\LARGE or}] (1) at (18, 6.5){};

\node[label=right:{(c)}] (1) at (20, 6.5){};
\draw [-stealth, orange, line width=2pt](25,7) -- (27,5.3);
\draw [-stealth, orange, line width=2pt](27,5) -- (26,3.3);
\node[fill,circle,inner sep=0pt,minimum size=\sizeS pt,label=right:{$u$}] (1) at (25, 7){};
\draw [-stealth](25,7) -- (23,5.3);

\node[fill,circle,inner sep=0pt,minimum size=\sizeS pt,label=right:{$v$}] (1) at (23.0, 5){};
\draw [-stealth](23,5) -- (22,3.3);
\draw [-stealth](23,5) -- (24,3.3);
\node[fill,circle,inner sep=0pt,minimum size=\sizeS pt] (1) at (27.0, 5){};

\draw [-stealth](27,5) -- (29,3.3);

\node[fill=blue,circle,inner sep=0pt,minimum size=\sizeL pt] (1) at (22, 3){};
\node[fill=red,circle,inner sep=0pt,minimum size=\sizeL pt] (1) at (24, 3){};

\node[fill=red,circle,inner sep=0pt,minimum size=\sizeL pt,label=right:{$t$}] (1) at (26, 3){};
\node[fill,circle,inner sep=0pt,minimum size=\sizeS pt,label=right:{$w$}] (1) at (29, 3){};
\draw [-stealth](29,3) -- (28,1.3);
\draw [-stealth](29,3) -- (30,1.3);

\node[fill=red,circle,inner sep=0pt,minimum size=\sizeL pt] (1) at (28, 1){};
\node[fill=blue,circle,inner sep=0pt,minimum size=\sizeL pt] (1) at (30, 1){};

\node[label=right:{(d)}] (1) at (31, 6.5){};
\draw [-stealth, orange, line width=2pt](36,7) -- (38,5.3);
\draw [-stealth, orange, line width=2pt](38,5) -- (40,3.3);
\draw [-stealth, orange, line width=2pt](40,3) -- (41,1.3);
%\node[label=right:{\LARGE or}] (1) at (28, 6.5){};
\node[fill,circle,inner sep=0pt,minimum size=\sizeS pt,label=right:{$u$}] (1) at (36, 7){};
\draw [-stealth](36,7) -- (34,5.3);

\node[fill,circle,inner sep=0pt,minimum size=\sizeS pt,label=right:{$v$}] (1) at (34.0, 5){};
\draw [-stealth](34,5) -- (33,3.3);
\draw [-stealth](34,5) -- (35,3.3);
\node[fill,circle,inner sep=0pt,minimum size=\sizeS pt,label=right:{$w$}] (1) at (38.0, 5){};
\draw [-stealth](38,5) -- (37,3.3);

\node[fill=blue,circle,inner sep=0pt,minimum size=\sizeL pt] (1) at (33, 3){};
\node[fill=red,circle,inner sep=0pt,minimum size=\sizeL pt] (1) at (35, 3){};

\node[fill=red,circle,inner sep=0pt,minimum size=\sizeL pt] (1) at (37, 3){};
\node[fill,circle,inner sep=0pt,minimum size=\sizeS pt] (1) at (40, 3){};

\draw [-stealth](40,3) -- (39,1.3);

\node[fill=blue,circle,inner sep=0pt,minimum size=\sizeL pt] (1) at (39, 1){};
\node[fill=red,circle,inner sep=0pt,minimum size=\sizeL pt,label=right:{$t$}] (1) at (41, 1){};

%\node[label=right:{\large all possible refinements of $T$ }] (1) at (14, -1.0){};
%\node[label=right:{\large where $e$ is witness of the new leaf~$t$}] (1) at (13, -2.5){};

%\node[label=right:{\large \textbf{\textcolor{orange}{orange}} path: Classification path of example~$e$ }] (1) at (-5, -5.2){};
\end{tikzpicture}
}
\caption{Examples of one-step refinements. 
(a) shows a current witness tree~$W$ where one red example~$e$ is misclassified in leaf~$z$.
(b), (c), and (d) show all three possible one-step refinements of~$W$ where~$e$ is the witness of the new leaf~$t$.
The orange path is the classification path of example~$e$.
Note that the new node can have any cut separating example~$e$ from the witness of leaf~$z$.}
\label{fig-os-ref}
\end{figure}

A one-step refinement~$(v,i,t,e)$ is applied to a witness tree~$W$ as follows (see \Cref{fig-os-ref}):
First, subdivide the edge from~$v$ to its parent~$p$ % (the case that~$v$ is the root is handled below)
by adding a new inner vertex~$u$ with cut~$(i,t)$.
Now,~$p$ is the parent of~$u$ and~$u$ is the parent of~$v$. 
If~$v$ was the left (right) child of~$p$ then~$u$ becomes the left (right) child of~$p$. 
%Similarly, if~$v$ was previously the right child of~$p$ then~$u$ is now the new right child. 
If~$v$ is the root of the tree (hence~$v$ has no parent~$p$) then~$u$ becomes the new root.
Second, add a new leaf~$\ell$ as the second child of~$u$ (with~$v$ being the first child).
The example~$e$ determines the assignment of~$\ell$ and~$v$ to the sides of the cut~$(i,t)$:
The example~$e$ needs to be assigned to leaf~$\ell$, so if~$e$ is assigned to the left child, then~$\ell$ becomes the left child. Otherwise,~$\ell$ becomes the right child. 
Finally, the class of~$\ell$ is set to~$\lambda(e)$ and the witness of~$\ell$ is set to~$e$. 

Let~$R$ denote the newly created witness tree.
We write~$W\xrightarrow{r} R$ to indicate that~$R$ was created by applying~$r$ to~$W$.
An application of a sequence~$r_1, \ldots, r_q$ of one-step refinements to create~$R$ is denoted by~$W\xrightarrow{r_1,\ldots,r_n} R$.
$\refine(W)$ is the set of all one-step refinements~$r = (v,i,t,e)$ of~$W$ such that~$e\in \dirty(W)$ is dirty in~$W$,~$e$ and~$\wit(\leaf(e))$ are on different sides of the cut~$(i,t)$ and~$r$ does \emph{not} change the leaf of any witness of~$W$.
This set is important because it contains all one-step refinements that are relevant for the algorithm.

\paragraph{Description of the algorithm solving \textsc{BSDT}.}
%We now describe the algorithm for solving BSDT; 
\Cref{alg:base} shows the pseudocode. For the correctness proof, refer to the work of \citet{KKSS23}.
\Cref{alg:base} starts with a witness tree~$W$ that consists of just one leaf~$\ell$.
An arbitrary witness for this leaf~$\ell$ is chosen and the class of~$\ell$ is set to the class of this witness.
Next, \texttt{Refine} is called (\Cref{line-refine}).
\Cref{alg:base} then recursively chooses a dirty example~$e$ and iterates over all one-step refinements in $\refine(W)$ in which~$e$ is the dirty example. The exact order of this iteration is specified in the appendix.
The idea is that since~$e$ is currently dirty, we need to assign~$e$ to a new leaf with class~$\lambda(e)$ since all examples are required to be correctly classified.
The one-step refinements with~$e$ as dirty example do this by making~$e$ the witness of the newly added leaf and assigning the class of~$e$ to that leaf.
\Cref{alg:base} traverses a search tree where each node represents a call of \texttt{Refine}.
To avoid confusion, from now on we call the vertices of the search tree \emph{nodes}, and the vertices in a decision/witness tree \emph{vertices}. For a node~$N$, we let~$\tree(N)$ denote the witness tree~$W$ that \texttt{Refine} is called with and we use~$\ex(N)$ to refer to the dirty example that is chosen in \Cref{line-dirty}.

\begin{algorithm2e}[t]
  \DontPrintSemicolon
  \SetKwProg{Fn}{Function}{}{}
  \SetKwFunction{Refine}{Refine}

    \KwIn{A witness tree $W$, a data set $(E, \lambda)$, and a maximum size $s\in \mathds{N}$.}

    \KwOut{A perfect witness tree of size at most $s$ or $\bot$ if none could be found.}

    \BlankLine
    \Fn{\Refine($W$, $(E,\lambda)$, $s$)\label{line-refine}}{

    \lIf{$W$ classifies $(E,\lambda)$}{\Return{$W$}}

    \lIf{$W$ has size $s$}{\Return{$\bot$}}

    $e \gets$ some dirty example from $\dirty(W)$\label{line-dirty}\;

      \ForAll{$r = (v, i, t, e)\in \refine(W)$\label{line-loop-osr}}{
        Apply $r$ to $W$ to create a new witness tree $R$\label{line-apply-osr}\;
        $R' \gets$ \Refine($R, (E, \lambda), s$)\label{line-apply-refine}\;
        \lIf{$R' \not= \bot$}{\Return{$R'$}}
      }
    \Return{$\bot$}
  }  
  \caption{Base witness-tree algorithm.}
  \label{alg:base}
\end{algorithm2e}

\toappendix{

\paragraph{Order of one-step refinements in \cref{line-loop-osr}.}
One of our algorithmic improvements (the Subset Constraints presented in \Cref{sec:subset_constraints}) relies on a specific order in which the one-step refinements are considered in \Cref{line-loop-osr}:
In the outermost loop, the algorithm fixes the vertex~$v$ of the one-step refinement by iterating over the vertices on the leaf-to-root path starting at the leaf~$\ell$ of~$e$.
In the next loop, the algorithm fixes the dimension~$i$ by iterating over all $\delta$~dimensions in which~$e$ and~$\wit(\ell)$ differ. 
In the innermost loop, the algorithm fixes the threshold~$t$ by iterating over all thresholds in~$\Thr(i)$ between~$e[i]$ and~$\wit(\ell)[i]$.
The algorithm always starts with the thresholds that are closest to~$e[i]$ regardless of whether~$e[i]>\wit(\ell)[i]$ or~$e[i]<\wit(\ell)[i]$. 
Any one-step refinement~$(v,i,t,e)$ that changes the leaf of some witness is not contained in~$\refine(W)$ and thus skipped.
% Note that the number of one-step refinements for each dirty example~$e$ is upper-bounded by~$\delta\cdot D\cdot s$~\citep{KKSS23}.\footnote{\citet{KKSS23} provided a bound of~$6\cdot\delta\cdot D\cdot s$ for the number of one-step refinements.
% A closer inspection, however, shows that the factor~$6$ can be omitted.}
% Since the recursion depth is bounded by~$s$, this gives the running time bound of~$\Oh((\delta Ds)^s\cdot sn)$.

}

We now describe some first improvements to this algorithm, leading to the first version of the solver. 

\paragraph{Dirty example priority.}
\label{sec:dirty_priority}

For the correctness it is not relevant which dirty example is chosen in \Cref{line-dirty}.
This permits our first improvement: choose a dirty example that minimizes the number of branches.
Ideally one can chose the dirty example such that the number of one-step refinements is as small as possible.
Computing this for every dirty example for every call of \texttt{Refine} takes too long, however.
Hence, we update this number for an example~$e$ only whenever~$e$ is assigned to a new leaf; 
%This can of course lead to these numbers being inaccurate if many one-step refinements are performed on the tree that do not change the leaf of~$e$.
preliminary experiments showed that this is a good trade-off.
We can use the same idea to choose the witness of the initial leaf and the initial dirty example.
%These two examples determine how many one-step refinements the algorithm has to go through in the initial call of \texttt{Refine}.
That means before the algorithm starts we can find the pair of examples with different classes that minimizes the number of one-step refinements, that is, the number of cuts separating them and choose one of these two elements as initial witness.

% \subparagraph{Solving \textsc{MSDT}.}
% \label{sec:solving_MSDT}

\paragraph{Data reduction.}

%%%%%%%%%%%%%%%%%%%%%%%%%%%%%%%%%%%%%%%%%%%%%%%%%%%%%%%%%%%%%%%%%%%%%%%%%%%%%%%%%%%%%%%%%%%%
%%%%%%%%%%%%%%%%%%%%%%%%%%%%%%%%%%%%%%%%%%%%%%%%%%%%%%%%%%%%%%%%%%%%%%%%%%%%%%%%%%%%%%%%%%%%
%% ONLY HERE FOR SHORT VERSION 
%%%%%%%%%%%%%%%%%%%%%%%%%%%%%%%%%%%%%%%%%%%%%%%%%%%%%%%%%%%%%%%%%%%%%%%%%%%%%%%%%%%%%%%%%%%%
%%%%%%%%%%%%%%%%%%%%%%%%%%%%%%%%%%%%%%%%%%%%%%%%%%%%%%%%%%%%%%%%%%%%%%%%%%%%%%%%%%%%%%%%%%%%

% \todo[inline]{Since we do not give the correctness proofs, we may as well move the definition of correctness to the supplementary material?}
In the following, we briefly describe the used rules.
Full correctness proofs and their experimental evaluation can be found in the appendix.
First, if there are two examples~$e_1$ and~$e_2$ having the same value in each dimension (recall that~$e_1$ and~$e_2$ are required to have the same class label), we remove one of them from~$(E, \lambda)$.
Second, if there is a dimension~$i$ such that all examples have the same value in this dimension, we remove~$i$. 
Third, if the instance has two equivalent cuts, then remove one of them.
Here, two cuts~$(i_1, t_1), (i_2, t_2)\in \Cuts(E)$ are \emph{equivalent} if~$E[\leq (i_1, t_1)] = E[\leq (i_2, t_2)]$.
For the fourth rule, consider a pair of cuts~$(i, t_1)$ and~$(i, t_2)$ in the same dimension~$i$ with $t_1 < t_2$. 
Now, if all examples on the left (right) side of both cuts have the same class, then remove~$(i, t_1)$ (or~$(i,t_2)$ respectively).
For the last rule, consider a pair of dimensions~$i_1$ and~$i_2$. 
If there is an ordering of the examples such that the values of the examples do not decrease in both dimensions, then one can construct a new dimension~$i_{1,2}$ such that for each cut in dimensions~$i_1$ and~$i_2$, there is an equivalent cut in dimension~$i_{1,2}$ and vice versa. 
Now,~$i_1$ and~$i_2$ are replaced by~$i_{1,2}$.

%%%%%%%%%%%%%%%%%%%%%%%%%%%%%%%%%%%%%%%%%%%%%%%%%%%%%%%%%%%%%%%%%%%%%%%%%%%%%%%%%%%%%%%%%%%%
%%%%%%%%%%%%%%%%%%%%%%%%%%%%%%%%%%%%%%%%%%%%%%%%%%%%%%%%%%%%%%%%%%%%%%%%%%%%%%%%%%%%%%%%%%%%
%% HERE EXTENDEND VERSION
%%%%%%%%%%%%%%%%%%%%%%%%%%%%%%%%%%%%%%%%%%%%%%%%%%%%%%%%%%%%%%%%%%%%%%%%%%%%%%%%%%%%%%%%%%%%
%%%%%%%%%%%%%%%%%%%%%%%%%%%%%%%%%%%%%%%%%%%%%%%%%%%%%%%%%%%%%%%%%%%%%%%%%%%%%%%%%%%%%%%%%%%%

\toappendix{
Here, we provide a more detailed description of all of our reduction rules and prove their correctness.
Furthermore, we provide examples for their application.
Formally, a reduction rule is the following.

%In this section, we will present some techniques for reducing the input data. The idea is that we can sometimes transform a dataset into a smaller dataset without changing the solution to any instance of \textsc{BSDT}. Here, smaller can refer to the number of examples, the number of dimensions, or the number of cuts. To formalize this idea we introduce reduction rules with the following definition.

\begin{definition}
	A \emph{reduction rule} is a function~$r$ that maps any data set~$(E, \lambda)$ to another data set~$(E', \lambda')$. We call~$r$ \emph{correct} if there is a decision tree of size~$s$ that correctly classifies~$(E, \lambda)$, if and only if there is a decision tree of size~$s$ that correctly classifies~$(E', \lambda')$.
\end{definition}

In the following subsections we introduce five reduction rules that decrease the size of the data set. To demonstrate how they work we use the data set in \Cref{tab:ex_preprocessing} as an example.
The overall effect of the application of all of our reduction rules is shown in \Cref{tab:overview_datasets_reduced}.

\subsection{Removing examples and dimensions}

The first two reduction rules are very simple. We can use them to remove unnecessary examples or dimensions.

\begin{rrule}[Remove Duplicate Example Rule]
Let~$(E,\lambda)$ be a data set and let~$e_1,e_2\in E$ be a pair of examples. If~$e_1$ and~$e_2$ have the same value in all dimensions, then remove~$e_1$.
\end{rrule}

\begin{rrule}[Remove Dimension Rule]
	Let~$(E,\lambda)$ be a data set and let~$i$ be a dimension. If all examples in~$i$ have the same value, then remove dimension~$i$.
\end{rrule}

The first rule is clearly correct because two examples with the same values in all dimensions have the same class and would always end up in the same leaf of any decision tree. Removing one of them therefore does not change whether a decision tree is perfect or not.

The second rule is clearly correct because a dimension where all examples have the same value does not have any cuts. That means a decision tree never use this dimension anyway.

\subsection{Equivalent cuts}

The third reduction rule is called the \emph{Equivalent Cuts Rule}. The idea of this rule is to remove cuts from the data set that are equivalent to other cuts. To better understand what it means for two cuts to be equivalent we can look at the two cuts~$(d_1,1)$ and~$(d_2,1)$ in \Cref{tab:ex_preprocessing}. Their left sides~$E[\leq (d_1,1)]$ and~$E[\leq (d_2,1)]$ are both equal to~$\{ a,b \}$. That means if we replace the cut~$(d_1,1)$ in a decision tree with the cut~$(d_2,1)$ no example would be assigned to a different leaf.

With this we can define an equivalence relation over the set of all cuts~$\Cuts(E)$. Two cuts~$(i_1,t_1)$,~$(i_2,t_2)\in \Cuts(E)$ are \emph{equivalent} if~$E[\leq (i_1,t_1)] = E[\leq (i_2,t_2)]$. We can now use this relation to define the Equivalent Cuts Rule:

\begin{rrule}[Equivalent Cuts Rule]
	Let~$(E,\lambda)$ be a data set and let~$(i_1,t_1), (i_2,t_2)\in \Cuts(E)$ be two cuts. If $(i_1,t_1)$ and $(i_2,t_2)$ are equivalent, then remove~$(i_1,t_1)$.
\end{rrule}

Next we explain how removing a cut~$(i,t)$ from a data set works. To remove this cut, we need to remove the thresholds~$t$ from~$\Thr(i)$. According to the definition of~$\Thr(i)$ there must be at least one pair of examples~$e_1,e_2\in E$ such that~$t$ is the only threshold in dimension~$i$ with~$e_1[i]\leq t<e_2[i]$. By setting the value in dimension~$i$ of the examples in every pair with this property to~$t$, we remove~$t$ from~$\Thr(i)$. At the same time this transformation ensures that the left and right side of every other cut does not change.

\begin{table}[t]
	\centering
	\begin{tabular}{ccccc}
		\toprule
		    & $d_1$ & $d_2$ & $d_3$ & class \\
		\midrule
		$a$ & $0$  & $1$  & $0$  & \textcolor{red}{red} \\
		$b$ & $1$  & $0$  & $0$  & \textcolor{red}{red} \\
		$c$ & $2$  & $2$  & $2$  & \textcolor{blue}{blue} \\
		$d$ & $3$  & $2$  & $1$  & \textcolor{red}{red} \\
		\bottomrule 
	\end{tabular}
	\caption{Example data set for the reduction rules.}
	\label{tab:ex_preprocessing}
\end{table}

If for example we want to remove the cut~$(d_1,1)$ from the data set in \Cref{tab:ex_preprocessing} we would find that the only pair of examples that has the property mentioned above is the pair~$(b,c)$. 
All other pairs can also be split by the cuts~$(d_1,0)$ or~$(d_1,2)$. That means we now set the values of~$b$ and~$c$ in dimension~$d_1$ to~$1$.

Now we just need to show that the Equivalent Cuts Rule is correct.

\begin{lemma}
	The Equivalent Cuts Rule is correct.
\end{lemma}
\begin{proof}
	Let~$(E,\lambda)$ be the original data set and~$(E',\lambda')$ be the data set that was created by the equivalent cuts rule. If a decision tree correctly classifies~$(E',\lambda')$ then it also correctly classifies~$(E,\lambda')$ since the rule only removes cuts. If a decision tree correctly classifies~$(E,\lambda)$ we can replace all cuts that were removed by the rule with the one cut that was not removed from their equivalence class. This creates a tree with the same size that correctly classifies~$(E',\lambda')$.
\end{proof}

\subsection{Reducing the size of dimensions}

With the fourth reduction rule, we want to remove the extreme values in some dimensions. For this let us again look at dimension~$d_1$ in \Cref{tab:ex_preprocessing}. The left side~$E[\leq (d_1,0)]$ of the cut~$(d_1,0)$ only contains a red example. The left side~$E[\leq (d_1,1)]$ of the cut~$(d_1,1)$ also only contains red examples while the right side contains fewer examples.

If one of these cuts is used in a decision tree that correctly classifies the data and does not contain empty leafs, then the left subtree of the vertex using these cuts would just be a single leaf with the class \textit{red}. That means replacing the cut~$(d_1,0)$ with the cut~$(d_1,1)$ in such a tree would create a decision tree that has the same size and still correctly classifies the data.

This leads to the following reduction rule.

\begin{rrule}[Dimension Reduction Rule]
	Let~$(E,\lambda)$ be a data set and let~$(i,t_1), (i,t_2)$ with~$t_1< t_2$ be a pair of cuts. If all examples on the left side of both cuts have the same class, then remove~$(i,t_1)$. Similarly, if all examples on the right side of both cuts have the same class, then remove~$(i,t_2)$.
\end{rrule}

Removing a cut works the same way as it does with the Equivalent Cuts Rule.

\begin{lemma}
	The Dimension Reduction Rule is correct.
\end{lemma}
\begin{proof}
	Let~$(E,\lambda)$ be the original data set and~$(E',\lambda')$ be the data set that was created by the equivalent cuts rule. If a decision tree correctly classifies~$(E',\lambda')$, then it also correctly classifies~$(E,\lambda')$ since the rule only removes cuts.
	
	Let~$\mathcal{D}$ be a decision tree that correctly classifies~$(E,\lambda)$. If the cut~$(i,t)$ of an inner vertex in~$\mathcal{D}$ was removed and the threshold~$t$ is smaller than the smallest threshold~$t'$ of dimension~$i$ in~$(E',\lambda')$, we can replace~$(i,t)$ with~$(i,t')$. We can also do this replacement if~$t'$ is the biggest threshold of dimension~$i$ in~$(E',\lambda')$ and~$t$ is bigger than~$t'$.
	
	Without loss of generality, let us assume that~$t$ is smaller than~$t'$. Since the rule removed~$(i,t)$ we know that all examples on the left sides of the cuts~$(i,t)$ and~$(i,t')$ must have the same class. Replacing~$(i,t)$ by~$(i,t')$ means the right side is still correctly classified while the left side can be classified without any further cuts. This means there is a tree with the same size as~$\mathcal{D}$ that correctly classifies~$(E',\lambda')$.
\end{proof}

\subsection{Merging dimensions}

With the last reduction rule we want to merge dimensions together in order to reduce the number of dimensions while keeping the number of different cuts the same. This may not seem useful at first but in \Cref{sec:subset_constraints} we introduce an improvement of the algorithm that benefits from this.

\begin{table}[t]
\centering
\begin{subfigure}{.49\textwidth}
	\centering
	\begin{tabular}{ccccc}
		\toprule
		    & $d_1$ & $d_2$ & $d_3$ & class \\
		\midrule
		$b$ & $1$  & $0$  & $0$  & \textcolor{red}{red} \\
		$a$ & $0$  & $1$  & $0$  & \textcolor{red}{red} \\
		$d$ & $3$  & $2$  & $1$  & \textcolor{red}{red} \\
		$c$ & $2$  & $2$  & $2$  & \textcolor{blue}{blue} \\
		\bottomrule
	\end{tabular}
	\caption{Before the merge.}
	\label{fig:ex_merge_dims_1}
\end{subfigure}
\begin{subfigure}{.49\textwidth}
	\centering
	\begin{tabular}{cccc}
		\toprule
		    & $d_1$ & $d_4$ & class \\
		\midrule
		$b$ & $1$  & $0$  & \textcolor{red}{red} \\
		$a$ & $0$  & $1$  & \textcolor{red}{red} \\
		$d$ & $3$  & $2$  & \textcolor{red}{red} \\
		$c$ & $2$  & $3$  & \textcolor{blue}{blue} \\
		\bottomrule
	\end{tabular}
	\caption{After the merge.}
	\label{fig:ex_merge_dims_2}
\end{subfigure}
\caption{Example for merging the two dimensions~$d_2$ and~$d_3$.}
\label{fig:ex_merge_dims}
\end{table}

To better demonstrate what merging two dimensions means we can look at dimensions~$d_2$ and~$d_3$ in \Cref{tab:ex_preprocessing}. If we now look at the examples in the order~$b,a,d,c$ as shown in \Cref{fig:ex_merge_dims_1} we can see that their values in both dimensions never decrease. We can now create a new dimension~$d_4$ with the values shown in \Cref{fig:ex_merge_dims_2}. We chose the values such that for each cut in dimension~$d_2$ or~$d_3$, there is an equivalent cut in~$d_4$. This means we can now completely remove dimensions~$d_2$ and~$d_3$ and if any decision tree uses a cut in one of these dimensions, we can just replace it with the equivalent cut in~$d_4$. This leads to the following rule.

\begin{table}[t!]
	\centering
	%\scalebox{0.75}{
	\small
	\begin{tabular}{lrrrrrrrrrr}
		\toprule
		Instance name & $n$ & $n'$ & $d$ & $d'$ & $c$ & $c'$ & $\delta$ & $\delta'$ & $D$ & $D'$ \\
		\midrule
		postoperative-patient-data & $72$ & $72$ & $22$ & $\mathbf{17}$ & $22$ & $22$ & $14$ & $14$ & $\mathbf{2}$ & $3$ \\
		hayes-roth & $84$ & $84$ & $15$ & $15$ & $15$ & $15$ & $8$ & $8$ & $2$ & $2$ \\
		lupus & $86$ & $\mathbf{79}$ & $3$ & $\mathbf{2}$ & $126$ & $\mathbf{78}$ & $3$ & $\mathbf{2}$ & $75$ & $\mathbf{53}$ \\
		appendicitis & $106$ & $106$ & $7$ & $7$ & $523$ & $\mathbf{460}$ & $7$ & $7$ & $99$ & $\mathbf{98}$ \\
		molecular\_biology\_promoters & $106$ & $106$ & $228$ & $228$ & $228$ & $228$ & $104$ & $104$ & $2$ & $2$ \\
		tae & $106$ & $106$ & $5$ & $5$ & $96$ & $\mathbf{94}$ & $5$ & $5$ & $46$ & $\mathbf{45}$ \\
		cloud & $108$ & $108$ & $7$ & $7$ & $585$ & $\mathbf{555}$ & $7$ & $7$ & $108$ & $\mathbf{100}$ \\
		cleveland-nominal & $130$ & $130$ & $17$ & $17$ & $17$ & $17$ & $11$ & $11$ & $2$ & $2$ \\
		lymphography & $148$ & $148$ & $50$ & $\mathbf{37}$ & $50$ & $50$ & $26$ & $\mathbf{23}$ & $\mathbf{2}$ & $3$ \\
		hepatitis & $155$ & $155$ & $39$ & $\mathbf{28}$ & $355$ & $\mathbf{335}$ & $28$ & $\mathbf{25}$ & $85$ & $85$ \\
		glass2 & $162$ & $162$ & $9$ & $9$ & $709$ & $\mathbf{667}$ & $9$ & $9$ & $136$ & $\mathbf{132}$ \\
		backache & $180$ & $180$ & $55$ & $\mathbf{50}$ & $469$ & $\mathbf{429}$ & $26$ & $26$ & $180$ & $\mathbf{151}$ \\
		auto & $202$ & $202$ & $52$ & $\mathbf{35}$ & $961$ & $\mathbf{916}$ & $31$ & $\mathbf{29}$ & $184$ & $\mathbf{182}$ \\
		glass & $204$ & $204$ & $9$ & $9$ & $894$ & $\mathbf{846}$ & $9$ & $9$ & $172$ & $\mathbf{165}$ \\
		biomed & $209$ & $209$ & $14$ & $14$ & $735$ & $\mathbf{577}$ & $9$ & $9$ & $191$ & $\mathbf{125}$ \\
		new-thyroid & $215$ & $\mathbf{214}$ & $5$ & $5$ & $329$ & $\mathbf{232}$ & $5$ & $5$ & $100$ & $\mathbf{73}$ \\
		spect & $219$ & $219$ & $22$ & $22$ & $22$ & $22$ & $22$ & $22$ & $2$ & $2$ \\
		breast-cancer & $266$ & $266$ & $31$ & $\mathbf{25}$ & $40$ & $40$ & $15$ & $15$ & $11$ & $11$ \\
		heart-statlog & $270$ & $270$ & $25$ & $\mathbf{23}$ & $376$ & $\mathbf{369}$ & $18$ & $18$ & $144$ & $\mathbf{142}$ \\
		haberman & $283$ & $283$ & $3$ & $3$ & $89$ & $\mathbf{86}$ & $3$ & $3$ & $49$ & $\mathbf{46}$ \\
		heart-h & $293$ & $293$ & $29$ & $\mathbf{22}$ & $325$ & $\mathbf{318}$ & $19$ & $19$ & $154$ & $154$ \\
		hungarian & $293$ & $293$ & $29$ & $\mathbf{22}$ & $325$ & $\mathbf{318}$ & $19$ & $19$ & $154$ & $154$ \\
		cleve & $302$ & $302$ & $27$ & $\mathbf{25}$ & $390$ & $\mathbf{382}$ & $18$ & $18$ & $152$ & $\mathbf{151}$ \\
		heart-c & $302$ & $302$ & $27$ & $\mathbf{25}$ & $390$ & $\mathbf{382}$ & $18$ & $18$ & $152$ & $\mathbf{151}$ \\
		cleveland & $303$ & $303$ & $27$ & $\mathbf{25}$ & $391$ & $\mathbf{383}$ & $18$ & $18$ & $152$ & $\mathbf{151}$ \\
		ecoli & $327$ & $\mathbf{326}$ & $7$ & $\mathbf{5}$ & $351$ & $\mathbf{233}$ & $6$ & $\mathbf{5}$ & $81$ & $\mathbf{59}$ \\
		schizo & $340$ & $340$ & $14$ & $14$ & $2218$ & $\mathbf{2209}$ & $14$ & $14$ & $203$ & $203$ \\
		bupa & $341$ & $341$ & $5$ & $5$ & $307$ & $\mathbf{302}$ & $5$ & $5$ & $94$ & $94$ \\
		colic & $357$ & $357$ & $75$ & $\mathbf{71}$ & $408$ & $\mathbf{400}$ & $36$ & $36$ & $85$ & $\mathbf{82}$ \\
		dermatology & $366$ & $366$ & $129$ & $\mathbf{101}$ & $188$ & $188$ & $57$ & $\mathbf{53}$ & $61$ & $61$ \\
		cars & $392$ & $\mathbf{388}$ & $12$ & $\mathbf{11}$ & $704$ & $\mathbf{531}$ & $9$ & $9$ & $346$ & $\mathbf{266}$ \\
		soybean & $622$ & $622$ & $133$ & $\mathbf{73}$ & $133$ & $\mathbf{108}$ & $68$ & $\mathbf{49}$ & $\mathbf{2}$ & $7$ \\
		australian & $690$ & $690$ & $18$ & $\mathbf{16}$ & $1155$ & $\mathbf{1119}$ & $16$ & $\mathbf{15}$ & $350$ & $350$ \\
		diabetes & $768$ & $768$ & $8$ & $8$ & $1246$ & $\mathbf{1238}$ & $8$ & $8$ & $517$ & $\mathbf{515}$ \\
		contraceptive & $1358$ & $1358$ & $21$ & $21$ & $66$ & $\mathbf{65}$ & $13$ & $13$ & $34$ & $34$ \\
		\bottomrule
	\end{tabular}
	%}
	\caption{Overview of the data sets we used for our experiments including their name, number of examples~$n$, number of dimensions~$d$, number of total cuts~$c$, maximum number~$\delta$ of dimensions in which two examples differ, and the largest domain size~$D$. The columns~$n', d', c', \delta'$, and~$D'$ show the values of the data sets after applying all reduction rules.\\
	\textbf{Bold} entries indicate a change of this specific value in the input instance and after the application of all reduction rules.}
	\label{tab:overview_datasets_reduced}
\end{table}

\begin{rrule}[Dimension Merge Rule]
	Let~$(E,\lambda)$ be a data set and let~$i_1,i_2$ be a pair of dimensions. If there is an ordering of the examples such that their values never decrease in either dimension, then we can create a new dimension~$i_{1,2}$ such that for each cut in dimensions~$i_1$ and~$i_2$, there is an equivalent cut in dimension~$i_{1,2}$. We then remove~$i_1$ and~$i_2$.
\end{rrule}

We can generate the values of each example in the new dimension in the following way: We go through the examples in the order mentioned in the rule. We then start by assigning each example the value~$0$. Every time the value of the examples increases in one of the two original dimensions, we need a cut in the new dimension that has all the previous examples on its left side. To achieve this, we increase the value we assign to the remaining examples in the new dimension by~$1$.

\begin{lemma}
	The Dimension Merge Rule is correct.
\end{lemma}
\begin{proof}
	Let~$(E,\lambda)$ be the original data set and~$(E',\lambda')$ be the data set that was created by the dimension merge rule. After adding the new dimension to~$(E',\lambda')$, there is an equivalent cut in this new dimension for each cut in the original two dimensions. The correctness of the Equivalent Cuts Rule shows that removing all cuts from the two dimensions that were merged together is correct. Removing these cuts leads to all examples having the same value in these two dimensions. The correctness of the Remove Dimension Rule shows that removing these two dimensions is also correct.
\end{proof}
}

\section{Lower bounds}
\label{sec:lower_bounds}
\appendixsection{sec:lower_bounds}

We introduce two lower bounds that we use to prune the search tree. 
In both lower bounds an instance of \textsc{Set Cover} is constructed and then a lower bound is calculated for this instance. 
In \textsc{Set Cover} the input is a universe~$U$ and a family~$S$ of subsets of~$U$, and the task is to compute the smallest integer~$k$ such that there is a subset~$S'\subseteq S$ of size exactly~$k$ whose union is the universe~$U$.
%The definition of~\textsc{Set Cover} is as follows.

%\probdef{Set Cover}
%{A universe~$U$ and a family~$S$ of subsets of~$U$.}
%{Compute the size~$k$ of the smallest subset~$S'\subseteq S$ such that the union of all sets in~$S'$ is equal to~$U$.}

We show that for our two lower bounds, the size~$k$ of the smallest subset~$S'$ is a lower bound for the minimum number of one-step refinements that are needed to correctly classify all examples in the current witness tree.
We calculate these lower bounds after Line~$5$ in each call of~\texttt{Refine} in \Cref{alg:base}. If~$k$ is bigger than~$s$ minus the current size of~$W$ we can return~$\bot$.
Since \textsc{Set Cover} is NP-hard~\citep{Karp72}, calculating the exact value of~$k$ in every call of \texttt{Refine} is not feasible. Instead we are going to introduce different ways of calculating a lower bound for~$k$.

\paragraph{First lower bound: Improvement Lower Bound (ImpLB).}

%The first lower bound is called the \emph{Improvement Lower Bound (ImpLB)}. 
The idea of the ImpLB is to find the minimum number of one-step refinements that are necessary to correctly classify only the examples that are dirty in the current witness tree~$W$ while ignoring the examples that are already correctly classified.
Moreover, we ignore that one-step refinements may interfere with each other and consider the effect of each one-step refinement separately.
To assess this effect, we use the following definition.

\begin{definition}\label{def:imp}
	Let~$W$ be a witness tree,~$E'\subseteq \dirty(W)$ a set of dirty examples and~$r\in \refine(W)$ a one-step refinement with~$W \xrightarrow{r} R$. Then we define
	\[\imp(W,E',r) \coloneqq \{ e'\in E' \mid e'\not\in\dirty(R) \}\]
	as the set of dirty examples that get correctly classified by~$r$. We call these sets the \emph{imp sets}.
\end{definition}

We now create an instance of \textsc{Set Cover} by choosing~$\dirty(W)$ as the universe and the family of subsets~$\{\imp(W,\dirty(W),r) \mid r\in \refine(W) \}$.
\toappendix{
\begin{figure}[t]
\centering
\begin{subfigure}{.49\textwidth}
	\centering
	\begin{tabular}{cccc}
		\toprule
		& $d_1$ & $d_2$ & class \\
		\midrule
		$a$ & $0$ & $3$ & \textcolor{blue}{blue} \\
		$b$ & $1$ & $2$ & \textcolor{red}{red} \\
		$c$ & $2$ & $2$ & \textcolor{blue}{blue} \\
		$d$ & $2$ & $1$ & \textcolor{red}{red} \\
		$e$ & $2$ & $0$ & \textcolor{blue}{blue} \\
		\bottomrule
	\end{tabular}
	\caption{An example data set with $n = 5$ and $d = 2$.}
	\label{fig:impLB_example_data}
\end{subfigure}
\begin{subfigure}{.49\textwidth}
	\centering
	\scalebox{1.3}{
	\begin{tikzpicture}
	
	\draw (0,0) -- (-1,-1);
	\draw (0,0) -- (1,-1);	
	
	\fill (0,0) circle (2pt) node[above=2pt] {$A$};
	\node at (0.8,0) {$(d_1,1)$};
	\fill (-1,-1) circle (2pt) [color=red] node[above=2pt,color=black] {$B$};
	\node at (-1,-1.3) {wit: $b$};
	\fill (1,-1) circle (2pt) [color=blue] node[above=2pt,color=black] {$C$};
	\node at (1,-1.3) {wit: $c$};
	
	\end{tikzpicture}
	}
	\caption{An example witness tree. $A$ has the cut $(d_1,1)$, $B$ has the witness $b$ and $C$ has the witness $c$.}
	\label{fig:impLB_example_tree}
\end{subfigure}
\caption{Example data set and witness tree.}
\label{fig:impLB_example}
\end{figure}
\subsection{An example for \cref{def:imp}}
To demonstrate how an instance like this can look, consider the example in \Cref{fig:impLB_example}. On the left it shows a data set and on the right it shows a witness tree~$W$ with one inner vertex~$A$ and two leafs~$B$ and~$C$. The set~$\refine(W)$ contains six one-step refinements.
\[\refine(W) = \{ (B,d_1,0,a), (B,d_2,2,a), (A,d_1,0,a), (A,d_2,2,a), (C,d_2,1,d), (A,d_2,1,d) \} \]
The imp sets~$\imp(W,\dirty(W),r)$ for these one-step refinements look like this.
\[ \{ a \}, \{ a \}, \{ a \}, \{ a \}, \{ d \}, \{ d \} \]
That means the \textsc{Set Cover} instance for this example has the universe~$U = \{ a,d\}$ and the set of subsets~$I = \{ \{ a\}, \{ d\} \}$.
}
%\subsubsection{Correctness Proof}
To show that the solution of this instance is a lower bound, we show that for any series of~$s$ one-step refinements~$I'$ that lead to~$W$ being perfect there is a set~$I\subseteq I'$ of size at most~$s$ such that the union of all sets in~$I$ is equal to~$\dirty(W)$.
%This is captured in Theorem~\ref{theorem:impLB}.

\begin{theorem}[\appref{theorem:impLB}]\label{theorem:impLB}
	Let~$W\coloneqq R_0$ be a witness tree and~$(r_1,\ldots,r_s)$ a series of one-step refinements with~$R_{i-1}\xrightarrow{r_i} R_i$, $r_i\in \refine(R_{i-1})$ for~$i\in [s]$ such that~$R_s$ is perfect. Then there must be a set~$I\subseteq \refine(W)$, $|I| \leq s$, such that \[\bigcup_{r\in I} \imp(W,\dirty(W),r) = \dirty(W).\]
\end{theorem}
\appendixproof{theorem:impLB}{
To prove \Cref{theorem:impLB} we first need to prove \Cref{lem:imp_subset}. 
With this lemma we want to show the following: If~$R$ was created by applying a one-step refinement to~$W$ then for any imp set $S$ in~$R$ there is an imp set in~$W$ that is a superset of $S$. An important detail is that we calculate these imp sets with respect to the dirty examples in~$W$. That means if there is a dirty example in~$R$ that was not dirty in~$W$ we ignore it.

\begin{lemma}\label{lem:imp_subset}
	Let~$W$ be a witness tree,~$E'\subseteq \dirty(W)$ a subset of the dirty examples in~$W$, and~$r\in \refine(W)$ a one-step refinement with~$W \xrightarrow{r} R$. Then, for each one-step refinement~$r_2\in \refine(R)$ there exists a one-step refinement~$r_1\in \refine(W)$, such that \[\imp(R,E'',r_2)\subseteq \imp(W,E',r_1)\] with~$E'' \coloneqq E' \setminus \imp(W,E',r)$.
\end{lemma}
\begin{proof}
	Let~$r = (v,i,t,e)$ and let~$r_2 = (v',i',t',e')\in \refine(R)$ be some one-step refinement for~$R$. We now just need to find some one-step refinement~$r_1\in \refine(W)$ such that $\imp(R,E'',r_2)\subseteq \imp(W,E',r_1)$. We can assume that~$\imp(R,E'',r_2)\not=\emptyset$, as otherwise~$\imp(R,E'',r_2) = \emptyset\subseteq\imp(W,E',r_1)$ for any~$r_1\in \refine(W)$. By definition we know that~$\imp(R,E'',r_2)\subseteq E''\subseteq E'\subseteq\dirty(W)$. Now we know there is some example~$e''\in \imp(R,E'',r_2)$ which is dirty in~$W$.
	
	Let~$u$ be the inner vertex and~$\ell$ the leaf that are added to~$W$ by~$r$. We can now distinguish between three cases based on the vertex~$v'$.
	\begin{enumerate}
		\item ($v'\in V(W)$): Here we choose~$r_1 = (v',i',t',e'')$. Since~$e''$ was correctly classified by~$r_2$ we know that~$e''$ must be in~$E[R,v']\subseteq E[W,v']$. We also know that~$r_1$ does not change the leaf of any witnesses since~$r_2\in\refine(R)$. Therefore we have~$r_1\in \refine(W)$. Any example~$d\in \imp(R,E'',r_2)$ is a dirty example in~$W$ and contained in~$E[W,v']$. Therefore~$r_1$ correctly classifies~$d$ and we have~$\imp(R,E'',r_2)\subseteq \imp(W,E',r_1)$.
		\item ($v' = u$): Here we choose~$r_1 = (v,i',t',e'')$. The proof for this case works as in the first case since we have~$E[R,v'] = E[R,u] = E[W,v]$.
		\item ($v' = \ell$): Here we choose~$r_1 = (v,i,t,e'')$. 
		We know that~$e''\in E[R,\ell]\subseteq E[W,v]$. 
		We also know that~$r_1$ does not change the leaf of any witnesses since~$r\in \refine(W)$. 
		Hence, we have~$r_1\in\refine(W)$. 
		Any example~$d\in \imp(R,E'',r_2)$ was assigned to~$\ell$ by the inner vertex~$u$ that was added by~$r$. 	
		Since we also have~$d\in E[W,v]$ we know that~$d$ is assigned to the new leaf that is created by~$r_1$. 
		All examples in~$\imp(R,E'',r_2)$ have the same class which means~$\lambda(e'') = \lambda(d)$. 
		Thus,~$r_1$ correctly classifies~$d$ and we have~$\imp(R,E'',r_2)\subseteq \imp(W,E',r_1)$.
	\end{enumerate}
\end{proof}

With this lemma we can now prove \Cref{theorem:impLB}.

\begin{proof}[Proof for \Cref{theorem:impLB}]
	
Let~$E_0 \coloneqq \dirty(W)$ and~$E_i \coloneqq E_{i-1}\setminus \imp(R_{i-1},E_{i-1},r_i)$ for all~$i\in [s]$. By \Cref{lem:imp_subset}, for each~$i\in [s]$ there is a one-step refinement~$r_i'\in \refine(W)$ with the following property:
\begin{equation}\label{eq:impLB_proof_1}
	\imp(R_{i-1},E_{i-1},r_i)\subseteq \imp(W,\dirty(W),r_i').
\end{equation}
Let~$I = \{ r_i' \mid i\in [s]\}$. Clearly~$|I|\leq s$. Since~$R_s$ classifies~$(E,\lambda)$ and~$E_s\subseteq \dirty(R_s)$ we must have~$E_s = \emptyset$. This implies~$\bigcup_{i\in [s]} \imp(R_{i-1},E_{i-1},r_i) = \dirty(W)$. Due to \Cref{eq:impLB_proof_1} we now also know~$\bigcup_{r\in I} \imp(W,\dirty(W),r) = \dirty(W)$.
\end{proof}
}

%\subsubsection{Calculating the ImpLB}
\paragraph{Calculating the ImpLB.}

%Since \textsc{Set Cover} is NP-hard~\citep{Karp72}, it is not feasible to calculate an optimal solution for this instance; instead we calculate a lower bound.
Instead of considering the content of each set, we only consider their sizes.
Thus, we calculate the smallest set of subsets such that the sum of their sizes is bigger or equal to the size of the universe. 
Hence, to calculate the ImpLB we just need to look at each one-step refinement~$r\in\refine(W)$, calculate the size of~$\imp(W,\dirty(W),r)$, sort the sizes in descending order, and then check how many sets are needed to reach a sum that is at least the size of~$\dirty(W)$.
We apply the ImpLB in each search-tree node.

\toappendix{
\subsection{Example for the ImpLB}
For the example in \Cref{fig:impLB_example} the way we calculate the ImpLB would not make a difference since each subset in~$I$ has size one and there is no overlap. But let us instead assume we had the universe~$U = \{ a, b, c, d\}$ and the set of subsets~$I = \{ \{ a, b\}, \{b, c\}, \{ d \} \}$. The optimal solution for this instance of \textsc{Set Cover} would be~$3$ since we need all three sets to cover~$U$. Our method turns~$I$ into a list of the sizes of each set sorted in descending order, giving~$[2, 2, 1]$. We would then sum up the numbers from left to right until we reach the size of~$U$. In this case we would need two numbers which is worse than the exact result. But for larger instances this method is much quicker to calculate since we do not need to look at the overlap between the sets.

\subsection{Improvements for the ImpLB}

We use two improvements to speed up this calculation while simultaneously improving the resulting lower bound.
First, we use known data reduction techniques for \textsc{Set Cover}: 
If a set~$A$ is a subset of a set~$B$ then we can remove~$A$ from our instance since it would always be better to choose~$B$ instead. 
Checking this subset relation for each pair in each call of ImpLB is too inefficient.
Thus, we just eliminate sets where we already know that they are a subset of a different set. We do this in the following way:
Consider two one-step refinements~$r_1 = (p,i,t,e)$ and~$r_2 = (v,i,t,e)$ where~$p$ is the parent of~$v$.
Clearly,~$\imp(W,\dirty(W),r_2)$ is a subset of~$\imp(W,\dirty(W),r_1)$ since both one-step refinements use the same cut and~$E[W,v]$ is a subset of~$E[W,p]$.
Thus, when calculating the ImpLB we only use one-step refinements~$r = (v,i,t,e)$ if the same one-step refinement can not be applied at the parent~$p$ of~$v$, i.e. when~$(p,i,t,e)$ is not in~$\refine(W)$.

%For the second improvement, observe that there can be many one-step refinements~$(\ell,i,t,e)$ for some fixed leaf~$\ell$.
%The sum of the corresponding imp sets might be much larger than the number~$x$ of dirty examples in~$\ell$.
%In this case, we only need to sum up the largest imp sets until at least~$x$ are reached.
%Clearly, this idea can be used for any subtree and not just leafs.

For the second improvement, observe that there can be many one-step refinements~$(v,i,t,e)$ for some fixed vertex~$v$.
The sum of the corresponding imp sets might be much larger than the number~$x$ of dirty examples in the subtree of~$v$.
By definition, each of these imp sets can only contain dirty examples that are contained in the subtree of~$v$.
Thus, we only need to consider the largest imp sets such that the sum of their sizes is at least equal to~$x$.

\subsection{Pseudocode for calculating the ImpLB}
\Cref{alg:impLB} shows the pseudo code for calculating the ImpLB as described above including the two improvements.
%\subparagraph{Algorithm for ImpLB}

\begin{algorithm2e}[t]
  \DontPrintSemicolon
  \SetKwProg{Fn}{Function}{}{}
  \SetKwFunction{CalcImpLB}{CalcImpLB}
  \SetKwFunction{Recurse}{Recurse}

    \KwIn{A witness tree $W$ and a data set $(E, \lambda)$.}

    \KwOut{A lower bound for the minimum number of one-step refinements that are still required for $W$ to classify $(E, \lambda)$.}

    \BlankLine
    
    \Fn{\CalcImpLB($W,(E, \lambda)$)}{
    
    \lForAll{$v\in V(W)$}{
    $P_v \gets$ empty Priority Queue
    }
    
    $root \gets$ root of $W$\;
    \Recurse($W, (E, \lambda), root$)\;
    \Return $|P_{root}|$\;
    }
    
    \Fn{\Recurse($W, (E, \lambda), v$)}{
    
    \If{$v$ is an inner vertex}{
    	$\ell, r \gets$ left child of $v$, right child of $v$\;
		\Recurse($W, (E, \lambda), \ell$)\;
		\Recurse($W, (E, \lambda), r$)\;
		$P_v.\text{addAll}(P_\ell)$\,\,\, and\,\,\, $P_v.\text{addAll}(P_r)$\;
    
    }
    $p \gets$ parent of $v$ or \texttt{null} if $v$ is the root of $W$\;
    
    \ForAll{$r = (v, i, t, e)\in \refine(W)$}{
    \If{$p =$ \texttt{null} $\vee\, (p, i, t, e)\notin \refine(W)$}{
		$P_v.\text{add}(|\imp(W,\dirty(W), r)|)$ \tcp*[f]{Improvement 1}\;
	}
    }
    $d \gets$ number of dirty examples in the subtree of $v$ \;
    
    \While{$\text{sum}(P_v) - P_v.\text{min}() > d$  }{
		$P_v.\text{removeMin}()$\tcp*[f]{Improvement 2}\;
	}
  }  
  \caption{ImpLB algorithm.}
  \label{alg:impLB}
\end{algorithm2e}
}

%The algorithm starts by initializing an empty priority queue for each vertex in the tree in Line~$2$. These are for storing and sorting the sizes of the imp sets. Next \textsc{Recurse} is called for the root of~$W$. \textsc{Recurse} starts by calling itself for the left and right child of~$v$ and adding the content of their priority queues to the priority queue of~$v$ in Lines~$7$ through~$12$. This is of course only done if~$v$ is not a leaf. Next, in Line~$14$, the algorithm iterates over every one-step refinement in~$\refine(W)$ that can be performed at~$v$. In Line~$15$ the algorithm checks that the same one-step refinement can not be applied at the parent if a parent exists. This is the first improvement we mentioned above. The size of the imp set of~$r$ is then added to~$P_v$ in Line~$16$. Finally in Lines~$17$ --~$19$ we remove the smallest numbers from~$P_v$ until removing even one more would make the sum of the numbers in~$P_v$ smaller than the number of dirty examples in the subtree of~$v$. This is the second improvement we mentioned above. The ImpLB is then just the size of~$P_{root}$ as can be seen in Line~$5$.

\paragraph{Second lower bound: Pair Lower Bound (PairLB).}
\label{sec:pairLB}

For the PairLB we use a similar idea as for the ImpLB but instead of looking at sets of dirty examples we look at sets of pairs of examples that are assigned to the same leaf but have different classes. For each one-step refinement~$r$ we define a \emph{pairsplit set} as the set of all pairs that are split up by $r$, that is, all pairs where the two examples are no longer in the same leaf of the tree created by $r$.
With the set of all pairs as the universe and all pairsplit sets as the family of subsets we obtain a~\textsc{Set Cover} instance. Similar to the ImpLB, the solution of this instance is a lower bound for the minimum number of one-step refinements that are needed to correctly classify all examples in the current witness tree.

However, preliminary experiments showed that calculating the PairLB in every search-tree node does not improve the running time of the algorithm. 
%If we calculate the PairLB the same way we calculate the ImpLB then the obtained lower bound is generally not good enough. On the other hand, methods that give us a better bound generally take too long to calculate.
Thus, we now only use the PairLB to calculate an initial lower bound for the solution of~\textsc{MSDT}.
For this we use the LP-relaxation of the standard ILP-formulation of~\textsc{Set Cover}. 
Further details  and proofs for the PairLB are in the Appendix.

\toappendix{
\subsection{Details for the Pair Lower Bound (PairLB)}
Next we define the \emph{Pair Lower Bound (PairLB)}. For this lower bound we look at pairs of vertices with different classes. For a witness tree~$W$, we define~$\pairs(W)$ as the \emph{set of all pairs}~$\{ e_1, e_2 \}\subseteq E$ with~$\lambda(e_1)\not= \lambda(e_2)$ and~$\leaf(e_1) = \leaf(e_2)$, that is the set of all pairs of examples that are assigned to the same leaf but have different classes. 
Clearly, each pair always contains exactly one dirty example. 
Furthermore, $\refineAll(W)$ is the set of \emph{all possible one-step refinements} of~$W$.
In particular~$\refine(W)\subseteq \refineAll(W)$.
%For the example in Figure~\ref{fig:impLB_example} this set would be~$\pairs(W) = \{ \{ a,b \}, \{ c,d \}, \{ d,e \} \}$.

In a perfect tree $W$ the set $\pairs(W)$ must be empty since examples with different classes must be in different leafs. Because of this we can now define a lower bound that is similar to the ImpLB but instead of correctly classifying all dirty examples we try to split up all existing pairs. We use \Cref{def:pairsplit} to capture this idea.

\begin{definition}\label{def:pairsplit}
	For some subset of pairs~$P\subseteq \pairs(W)$ and a one-step refinement~$r\in \refineAll(W)$ with~$W \xrightarrow{r} R$, we define
	\[\pairsplit(W,P,r) \coloneqq \{ p\in P \mid p\not\in\pairs(R) \}\]
	as the \emph{set of pairs that get split up by~$r$}.
        We call these sets the \emph{pairsplit sets}.
\end{definition}

Similar to the imp sets a pairsplit set is a subset of~$\pairs(W)$ in which all pairs can be split up by a single one-step refinement.

We can now create an instance of \textsc{Set Cover} by choosing~$\pairs(W)$ as the universe and~$P = \{ \pairsplit(W,\pairs(W),r) \mid r\in \refineAll(W) \}$ as the family of subsets. Notice how we use the set of all one-step refinements~$\refineAll(W)$ instead of the set~$\refine(W)$ that only contains one-step refinement that the algorithm is allowed to use. This is because this lower bound would not work if we used~$\refine(W)$ instead of~$\refineAll(W)$.

%\toappendix{
\paragraph{Example for the Pair Lower Bound (\Cref{def:pairsplit}).}
We can see this by looking at the example in \Cref{fig:impLB_example} again. If we construct the pairsplit sets for each one-step refinement in~$\refine(W)$ for this example we would get the following sets.
\[ \{ \{ a,b \} \}, \{ \{ a,b \} \}, \{ \{ a,b \} \}, \{ \{ a,b \} \}, \{ \{ c,d \} \}, \{ \{ c,d \} \}. \]

Notice how none of these sets contain the pair~$\{ d,e \}$. This means if we only used the one-step refinements in~$\refine(W)$ to construct the \textsc{Set Cover} instance, it would have no solution. If we actually use all one-step refinements in~$\refineAll(W)$ to construct the set of subsets~$P$ and we remove all duplicates, we would get~$P = \{ \{ \{ a,b \} \}, \{ \{ c,d \} \}, \{ \{ d,e \} \} \}$ since there are no one-step refinements that can split more than one pair at a time.
%}

\paragraph{Correctness proof.}

To show that the solution of such an instance of \textsc{Set Cover} is a correct lower bound, we just need to show that for any series of~$s$ one-step refinements that lead to~$W$ being perfect, there is a set~$P'\subseteq P$ of size at most~$s$ such that the union of all sets in~$P'$ is equal to~$\pairs(W)$.

\begin{theorem}%[\appref{theorem:pairLB}]
\label{theorem:pairLB}
	Let $W\coloneqq R_0$ be a witness tree and $(r_1,\ldots,r_s)$ a series of one-step refinements with $R_{i-1}\xrightarrow{r_i} R_i$, $r_i\in \refineAll(R_{i-1})$ such that $R_s$ is perfect. Then, there must be a set of at most $s$ one-step refinements $I\subseteq \refineAll(W)$, $|I| \leq s$, such that \[\bigcup_{r\in I} \pairsplit(W,\pairs(W),r) = \pairs(W).\]
\end{theorem}
%\appendixproof{theorem:pairLB}{
Again similar to the ImpLB, we first need a lemma that has the same function as \Cref{lem:imp_subset}: If $R$ was created by applying a one-step refinement to $W$ then for any pairsplit set in $R$, we want to show that there is a pairsplit set in $W$ that is a superset of the pairsplit set in $R$. An important detail is that we calculate these pairsplit sets with respect to $\pairs(W)$. That means if there is a pair in $\pairs(R)$ that is not in $\pairs(W)$, we ignore it.

\begin{lemma}\label{lem:pairsplit_subset}
	Let~$W$ be a witness tree,~$P\subseteq \pairs(W)$ a subset of the pairs in~$W$ and~$r\in \refine(W)$ a one-step refinement with~$W \xrightarrow{r} R$. Then for each one-step refinement~$r_2\in \refineAll(R)$, there exists a one-step refinement~$r_1\in \refineAll(W)$, such that~\[\pairsplit(R,P',r_2)\subseteq \pairsplit(W,P,r_1)\] with~$P' \coloneqq P \setminus \pairsplit(W,P,r)$.
\end{lemma}
\begin{proof}
	Let~$r_2 = (v,i,t,e)\in \refineAll(R)$ be some one-step refinement for~$R$. We now just need to find some one-step refinement~$r_1\in \refineAll(W)$ such that~$\pairsplit(R,P',r_2)\subseteq \pairsplit(W,P,r_1)$.
	
	Since~$R$ was created by applying a one-step refinement to~$W$ we have~$\refineAll(W)\subseteq \refineAll(R)$. Because of this we can simply choose~$r_1 = r_2$. We know that~$P'\subseteq P$ and we know that pairs in~$P'$ can never be added to subtrees in~$R$ that they were not originally in before~$r$ was applied to~$W$. Because of this we have~$\pairsplit(R,P',r_2)\subseteq \pairsplit(W,P,r_1)$.
\end{proof}

Now we can prove \Cref{theorem:pairLB}. Due to the similarities in how we defined the ImpLB and the PairLB this proof works exactly the same as the proof for \Cref{theorem:impLB}.

\begin{proof}[Proof for \Cref{theorem:pairLB}]
	
	Let~$P_0 \coloneqq \pairs(W)$ and~$P_i \coloneqq P_{i-1}\setminus \pairsplit(R_{i-1},P_{i-1},r_i)$ for all~$i\in [s]$. Due to \Cref{lem:pairsplit_subset} we know that for each~$i\in [s]$ there is a one-step refinement~$r_i'\in \refineAll(W)$ with the following property:
\begin{equation}\label{eq:pairLB_proof_1}
	\pairsplit(R_{i-1},P_{i-1},r_i)\subseteq \pairsplit(W,\pairs(W),r_i').
\end{equation}
Let~$I = \{ r_i' \mid i\in [s]\}$. Clearly~$|I|\leq s$. Since~$R_s$ is perfect and~$P_s\subseteq \pairs(R_s)$ we must have~$P_s = \emptyset$. That implies~$\bigcup_{i\in [s]} \pairsplit(R_{i-1},P_{i-1},r_i) = \pairs(W)$. Due to \Cref{eq:pairLB_proof_1} we now also know~$\bigcup_{r\in I} \pairsplit(W,\pairs(W),r) = \pairs(W)$.
\end{proof}
%}

\paragraph{Calculating the PairLB.}

To actually calculate a lower bound for this instance of \textsc{Set Cover} we could use the same method we used for the ImpLB. However our preliminary experiments showed that this method does not yield very good results if we use it to calculate the PairLB. This is probably due to the fact that we use~$\refineAll(W)$ instead of~$\refine(W)$.

Instead we translate the problem instance into an ILP instance. We can then create an LP-relaxation of that instance and solve this relaxation with an LP solver. Clearly the solution for the LP-relaxation is a lower bound for the solution of the ILP instance since we are dealing with a minimization problem.

To construct the ILP instance we use the standard ILP formulation for \textsc{Set Cover}. First, we introduce a binary variable~$x_r$ for every one-step refinement~$r\in\refineAll(W)$. We want $x_r$ to be equal to~$1$ if and only if the~$\pairsplit$ set of~$r$ is part of the set cover solution. Next we are going to introduce a constraint~$1\leq \sum_{r\in S_p} x_r$ for every pair~$p\in \pairs(W)$ where~$S_p \coloneqq \{ r\in \refineAll(W) \mid p\in\pairsplit(W,\pairs(W),r) \}$ is the set of all one-step refinements that split~$p$. The objective function is then just to minimize the sum of all variables. The LP-relaxation is obtained by allowing the variables to have any value between~$0$ and~$1$. This of course also makes it possible that the result of the objective function is not an integer. In that case we simply round up to the nearest integer.

To reduce the size of the LP instance and speed up the calculation we can again use one of the data reduction techniques that we used for the ImpLB. This time we can eliminate all subsets~$S$ that come from one-step refinements that are not applied to the root of the tree. Since we are using one-step refinements from~$\refineAll(W)$ we can always apply the same one-step refinement at the root of~$W$ and get a subset that is a superset of~$S$.
}

\section{Subset constraints}
\label{sec:subset_constraints}
\appendixsection{sec:subset_constraints}

%In this section we introduce subset constraints. 
A \emph{subset constraint of a vertex~$v$} in a witness tree~$W$ is a subset~$S\subseteq E[W,v]$ of examples which imposes the constraint that one-step refinements are not allowed to remove all examples of~$S$ from the subtree of~$v$.
The idea is that if such one-step refinements lead to a perfect tree, then there is a different perfect tree that does not violate this constraint and is not larger.
For example, if one can replace threshold~$t$ in vertex~$v$ of~$W$ by a different threshold~$t'$ such that there is no example in~$v$ with a threshold between~$t$ and~$t'$, we only have to test one of~$t$ and~$t'$.
Testing this for each pair of thresholds is too inefficient.
Instead, we use subset constraints: 

\begin{definition}\label{def:subset_constraint}
	Let~$W$ be a witness tree and let~$v\in V(W)$. 
	We call a subset~$C\subseteq E$ a \emph{Subset Constraint} of~$v$. We call~$C$ \emph{violated} if~$E[W,v]\cap C = \emptyset$. 
	The set~$\const(W,v)$ contains all subset constraints of~$v$ in the tree~$W$.
\end{definition}

We add a subset constraint~$C$ when~$v$ is first added to the tree% with the threshold~$t$
.
After each one-step refinement, we update which examples of~$C$ are still in the subtree of~$v$, thus immediately detecting a violation of~$C$.
%We formalize the idea of a subset constraint in the following definition.

%The intuition behind subset constraints is that sometimes a witness tree~$W$ can be transformed into a slightly different witness tree~$W'$ without increasing the size and without changing the classes that the tree assigns to the examples. This could for example be done by simply changing the threshold~$t$ of a vertex~$v$ in~$W$ to a different threshold~$t'$. This does not change the size of~$W$ and as long as there are no examples in the subtree of~$v$ that have values between~$t$ and~$t'$ then there is no example that gets assigned to a different leaf by this change. That means the algorithm can skip the search tree node~$N$ with~$\tree(N) = W$. If one of these search tree nodes leads to a correct tree then there would be a different search tree node that leads to the same tree where~$t$ is replaced by~$t'$.

%However, detecting that there are no examples inbetween~$t$ and~$t'$ by simply checking every pair of thresholds for every vertex is inefficient. 
%Instead we want to use a subset constraint~$C$ to keep track of the values inbetween~$t$ and~$t'$. 
%We would add~$C$ when~$v$ is first added to the tree with the threshold~$t$. 
%Then, every time a one-step refinement changes the tree we update which examples of~$C$ are still in the subtree of~$v$. That way we immediatly know when~$C$ is violated.

%To prove the correctness of the specific subset constraints we introduce later in this section, we first need to introduce the concept of a refinement.

\paragraph{Refinements.} 
A refinement of a witness tree~$W$ is defined as follows.

%We want to be able to show that a witness tree~$R$ was created by applying one or more one-step refinements to a witness tree~$W$ by looking at the structure of both trees. For this we look at the following properties of one-step refinements.

%First we can notice that a one-step refinement in~$\refine(W)$ never changes any vertices that already exist in~$W$, it only adds an inner vertex and a leaf. The role of a vertex also never changes: a leaf always be a leaf and an inner vertex always be an inner vertex. The relative positions of vertices to eachother do not change either. This means that if a vertex is in the left subtree of some vertex then it always be in that left subtree. The same is true for the vertices in the right subtree. Lastly we know that a witness always stay in its leaf. This leads to the following definition.

\begin{definition}\label{def:refinement}
	Let~$W = (T,\cut,\cla,\wit)$ and~$R = (T',\cut',\cla',\wit')$ be two witness trees. We say that~$R$ is a \emph{refinement} of~$W$ if and only if~$W$ and~$R$ fulfill the following properties:

	\begin{enumerate}%[itemsep=-1ex]
		\item \label{def:refinement_1} $V(W)\subseteq V(R)$.
		\item \label{def:refinement_2} A vertex $v\in V(W)$ is a leaf in $W$ if and only if~$v$ is a leaf in $R$, and for each leaf $\ell\in V(W)$ we have $\cla(\ell) = \cla'(\ell)$ and $\wit(\ell)\in E[R,\ell]$.
		\item \label{def:refinement_3} Let $v_1,v_2\in V(W)$ be a pair of vertices. The vertex $v_1$ is in the left subtree of $v_2$ in $W$ if and only if~$v_1$ is in the left subtree of $v_2$ in $R$. The vertex $v_1$ is in the right subtree of $v_2$ in $W$ if and only if~$v_1$ is in the right subtree of $v_2$ in $R$.
		%\item\label{def:refinement_4} For every leaf $\ell\in V(W)$ we have $\cla(\ell) = \cla'(\ell)$ and $\wit(\ell)\in E[R,\ell]$.
		\item \label{def:refinement_5} For every inner vertex $v\in V(W)$ we have $\cut(v) = \cut'(v)$.
	\end{enumerate}
\end{definition}

If~$R$ was created by the algorithm by applying one or multiple one-step refinements to~$W$ then~$R$ is a refinement of~$W$.
In the appendix, we show that also the reverse direction is true.
%Now we also want to show the other direction. For this we use the following lemma.

\toappendix{
In the main part we showed that if~$R$ was created by the algorithm by applying one or multiple one-step refinements to~$W$ then~$R$ is a refinement of~$W$.
Next, we verify the reverse direction.

\begin{lemma}%[\appref{lem:unique_refinement}]
\label{lem:unique_refinement}
	Let~$R$ be a perfect witness tree that is a refinement of the witness tree~$W$ and let~$e\in\dirty(W)$ be a dirty example in~$W$. Then, there must be a one-step refinement~$r\in\refine(W)$ with~$W\xrightarrow{r} W'$ such that~$R$ is a refinement of~$W'$.
\end{lemma}
%\appendixproof{lem:unique_refinement}{
\begin{proof}
	We need to find a one-step refinement~$r = (v,i,t,e)$ with~$W\xrightarrow{r} W'$ such that~$R$ is a refinement of~$W'$. That means the leaf~$\ell$ and the inner vertex~$u$ we add with~$r$ need to be in~$V(R)$ due to Property~\ref{def:refinement_1} in \Cref{def:refinement}. To identify these vertices we can use the dirty example~$e$. Let~$\ell'$ be the leaf that~$e$ is assigned to in~$W$. The one-step refinement~$r$ makes~$e$ the witness of~$\ell$. Property~\ref{def:refinement_2} tells us that~$\ell$ must be the leaf that~$e$ is assigned to in~$R$. The inner vertex~$u$ must now be the vertex in~$R$ that has~$\ell$ in one of its subtrees and~$\ell'$ in the other subtree. There can only be one vertex like this. This inner vertex~$u$ now tells us the cut~$(i,t)$ we need for~$r$ due to Property~\ref{def:refinement_5}. Now we just need to find the vertex~$v$ in~$V(W)$ at which~$r$  is applied. For this we can look at the path from~$u$ to~$\ell'$ in~$R$. Starting from~$u$, the vertex~$v$ must be the first vertex on that path that also exists in~$V(W)$.
\end{proof}
%}

For the correctness of \Cref{lem:unique_refinement} it is essential that~$W$ has at least one dirty example.
This is because otherwise there is always a smaller perfect tree.
This is \emph{not} a crucial restriction since our goal is to find the smallest perfect tree.
}

\paragraph{Definition and proof of subset constraints.}
Now, we add subset constraints to a witness tree after the application of a one-step refinement.

%Now we are going to properly define and prove the subset constraints that we are going to use to improve the algorithm. In order to show their correctness we are first going to introduce Definition~\ref{def:correct_sub_const}. This definition assumes that subset constraints can be added to a tree after a one-step refinement has been applied. We define the exact conditions for this later in this section when we introduce the specific subset constraints.

\begin{definition}\label{def:correct_sub_const}
	Let~$N$ be a node in the search tree with witness tree~$W$ and the dirty example~$e \coloneqq \ex(N)$ and let~$r \coloneqq (v,i,t,e)\in \refine(W)$ be a one-step refinement with~$W\xrightarrow{r} W'$ that introduces a subset constraint~$C$ for the newly added inner vertex.
	
	We say that~$C$ is \emph{correct} if for each perfect witness tree~$R$ that is a refinement of~$W'$ and violates~$C$, there is a different witness tree~$R'$ with the following properties:
	\begin{enumerate}%[itemsep=-1ex]
		\item\label{def:correct_sub_const_1} $R'$ is also perfect.
		\item\label{def:correct_sub_const_2} $R'$ is not bigger than $R$.
		\item\label{def:correct_sub_const_3} There is a different one-step refinement~$r' \in \refine(W)$ with~$W\xrightarrow{r'} W''$ that \Cref{alg:base} chooses before~$r$ such that~$R'$ is a refinement of~$W''$.
	\end{enumerate}
\end{definition}

\toappendix{

\subsection{Intuition of \Cref{def:correct_sub_const}}

The most important part of \Cref{def:correct_sub_const} is Property~\ref{def:correct_sub_const_3}, that is, that~$r'$ is chosen before~$r$. 
Together with the other properties, according to \Cref{lem:unique_refinement}, there must be a perfect witness tree~$R''$ that is not bigger than~$R$, of which~$R'$ is a refinement and which is discovered by the algorithm before~$R$.
That means, if there is a perfect tree~$W$ that violates a correct subset constraint there must be a different perfect tree~$W'$ that is discovered by the algorithm before~$W$ and is not bigger than~$W$. 
If~$W'$ now also violates a correct subset constraint we can find a perfect tree~$W''$ that is discovered by the algorithm before~$W'$ and is not bigger than~$W'$. 
We can keep doing this until we eventually find a perfect tree that does not violate any correct subset constraints. Due to \Cref{def:correct_sub_const} and there only being a finite number of trees, such a tree must always exist. 
}

\begin{theorem}[\appref{theorem:sub_const_correctness}]
\label{theorem:sub_const_correctness}
	If there is at least one perfect witness tree of size at most~$s\in\mathds{N}$, then there is at least one perfect witness tree of size at most~$s$ that does not violate any correct subset constraints.
\end{theorem}
\appendixproof{theorem:sub_const_correctness}{
\begin{proof}
	Let us assume there is a perfect witness tree~$W$ that has at most size~$s$ and violates at least one correct subset constraint~$C$. According to \Cref{def:correct_sub_const} there must be a different perfect witness tree~$W'$ that is not bigger than~$W$ and is discovered by the algorithm before~$W$.
	
	We can keep replacing the current tree with a different tree according to \Cref{def:correct_sub_const} until we find a tree where no correct subset constraint is violated. Since we know that the replacement tree is always discovered earlier by the algorithm and there is only a finite number of trees we must find such a tree eventually.
\end{proof}
}

Next, we introduce two specific subset constraints.

%With this we now introduce two specific subset constraints. To show that the algorithm can skip any tree where one of these subset constraints is violated we just need to show that they are correct.

\paragraph{Threshold Subset Constraints.}

Our first subset constraint is based on the observation: replacing a threshold by another threshold without changing the leaf assignment of any example can be possible.

%First we introduce a subset constraint that uses the idea from the example at the beginning of this section where we showed that we can sometimes replace a threshold~$t$ with a different threshold~$t'$ without changing the leaf that any example is assigned to.

\begin{definition}\label{def:threshold_sub_const}
	Let~$N$ be a node in the search tree with the witness tree~$W \coloneqq \tree(N)$ and the dirty example~$e \coloneqq \ex(N)$, let~$r=(v,i,t,e),r'=(v,i,t',e)\in \refine(W)$ with~$W\xrightarrow{r} R$ and~$W\xrightarrow{r'} R'$ such that we have~$e[i]\leq t' < t$ or~$t < t' < e[i]$, and let~$\ell$ and~$u$ be the leaf and inner vertex that are added to~$W$ by~$r$.
We add the \emph{Threshold Subset Constraint}~$C \coloneqq E[R,\ell]\setminus E[R',\ell]$ to~$\const(R,u)$ and~$t'$ is the \emph{Constraint Threshold} of~$C$. 
\end{definition}

If a Threshold Subset Constraint of a vertex~$u$ in~$W$ is violated then replacing the threshold of~$u$ by the Constraint Threshold~$t'$ does not change the leaf of any example in~$W$. 
%We now use this observation to show that a Threshold Subset Constraint is always correct.

\begin{theorem}[\appref{thm-threshold-sub-const}]
  Threshold Subset Constraints are correct.%
  \label{thm-threshold-sub-const}
\end{theorem}%
\appendixproof{thm-threshold-sub-const}{
\begin{proof}%
%	Let~$N$ be a node in the search tree of \Cref{alg:base} with the witness tree~$W \coloneqq \tree(N)$ and the dirty example~$e \coloneqq \ex(N)$, let~$(v,i,t,e),(v,i,t',e)\in \refine(W)$ with~$W\xrightarrow{(v,i,t,e)} R$ and~$W\xrightarrow{(v,i,t',e)} R'$ such that we have~$e[i]\leq t'< t$ or~$t < t' < e[i]$, and let~$\ell$ and~$u$ be the leaf and inner vertex that are added to~$W$ by~$(v,i,t,e)$ and let~$C\in \const(R,u)$ be the Threshold Subset Constraint added to~$u$ in~$R$.
	We use the notation of \Cref{def:threshold_sub_const}; so let~$C\in \const(R,u)$ be the Threshold Subset Constraint added to~$u$ in~$R$.
	Furthermore, let~$R^*$ be a perfect witness tree that is a refinement of~$R$ and violates~$C$. 
	We find a different witness tree that has the three properties from \Cref{def:correct_sub_const}. 
	We can create such a witness tree~$R^{**}$ by simply replacing the threshold of~$u$ in~$R^*$ by the Constraint Threshold~$t'$. 
	Since~$C$ is violated in~$R^*$ we know that~$R^{**}$ is a perfect tree. 
	We also have not increased the size of the tree. Now we just need to show Property~\ref{def:correct_sub_const_3} of \Cref{def:correct_sub_const}.
	
Because~$t'$ is closer to~$e[i]$ than~$t$, \Cref{alg:base} chooses~$r=(v,i,t',e)$ before~$r'=(v,i,t,e)$ (see \Cref{sec:base_algo}). 
The only difference between~$R$ and~$R'$ and between~$R^*$ and~$R^{**}$ is the threshold of~$u$. 
	Thus,~$R^{**}$ is a refinement of~$R'$.
\end{proof}
}

\paragraph{Dirty Subset Constraints.}

Our second subset constraint is based on the idea that we  discard some vertices on the leaf-to-root path for one-step refinements.
For example, we may discard one-step refinements~$r$ at a vertex~$v$ when one of the child subtrees of~$v$ is already perfect, since it is also possible to apply~$r$ at the root of the other child subtree.

\begin{definition}\label{def:dirty_sub_const}
	Let~$N$ be a node in the search tree of \Cref{alg:base} with the witness tree~$W \coloneqq \tree(N)$ and the dirty example~$e \coloneqq \ex(N)$, let~$v$ be an inner vertex in~$W$ with the children~$v_1$ and~$v_2$, let~$r=(v,i,t,e),r'=(v_1,i,t,e)\in \refine(W)$ be two one-step refinements with~$W\xrightarrow{r} R$ and~$W\xrightarrow{r'} R_1$, let~$\ell$ and~$u$ be the leaf and inner vertex that are added to~$W$ by~$r$.
	We add the \emph{Dirty Subset Constraint}~$C = E[W,v_2]\cap \dirty(W)$ to~$\const(R,u)$.
\end{definition}

\begin{theorem}[\appref{thm-dirty-sub-const-corr}]
	Dirty Subset Constraints are correct.
	\label{thm-dirty-sub-const-corr}
\end{theorem}
\appendixproof{thm-dirty-sub-const-corr}{
\begin{proof}
	Let~$N$ be a node in the search tree of \Cref{alg:base} with the witness tree~$W \coloneqq \tree(N)$ and the dirty example~$e \coloneqq \ex(N)$, let~$v$ be an inner vertex in~$W$ with the children~$v_1$ and~$v_2$ and let~$(v,i,t,e),(v_1,i,t,e)\in \refine(W)$ be two one-step refinements with~$W\xrightarrow{(v,i,t,e)} R$ and~$W\xrightarrow{(v_1,i,t,e)} R_1$. Finally, let~$\ell$ and~$u$ be the leaf and inner vertex that are added to~$W$ by~$(v,i,t,e)$.
	
	Next, let~$R^*$ be a perfect witness tree that is a refinement of~$R$ and violates~$C$. We find a different witness tree~$R^{**}$ that has the three properties from \Cref{def:correct_sub_const}. Without loss of generality, we assume that~$v_1$ is the left child of~$v$. We now obtain~$R^{**}$ in the following way.
	
	First we know that~$v$ and~$u$ still exist in~$R^*$ but~$u$ may no longer be the direct parent of~$v$. So we are now going to take the entire subtree of~$u$ with the exception of the subtree of~$v$, remove the edge from~$v$ to its left child and move the subtree of~$u$ into that gap. We connect everything so that the previous parent of~$u$ is now the parent of~$v$ and the previous parent of~$v$ is now the parent of the previous left child of~$v$. Lastly we replace the right subtree of~$v$ in~$R^{**}$ with the right subtree of~$v$ in~$W$.
	
	We know that~$E[R^{**},u]\subseteq E[R^*,u]$ which means that everything in the subtree of~$u$ is still correctly classified. We also know that~$E[R^*,v]\subseteq E[R^{**},v]$. But we can also see that~$E[R^{**},v]\subseteq E[W,v]$. We replaced the right subtree of~$v$ in~$R^{**}$ with the right subtree of~$v$ in~$W$ and we know that~$C$ is violated which means all dirty examples in that right subtree were removed. This means that~$R^{**}$ is perfect.
	
	These changes also do not increase the size of the tree and they also do not change a leaf into an inner vertex or an inner vertex into a leaf.
	If we now only look at the vertices that were already present in~$R$ we can notice that the relative position of most of these vertices has not changed. The only changes are that~$u$ and~$\ell$ are now in the left subtree of~$v$ instead of~$v$ being a child of~$u$. This is the exact same difference that we have between~$R$ and~$R_1$. That means~$R^{**}$ is a refinement of~$R_1$.
	
	Since the algorithm chooses~$(v_1,i,t,e)$ before~$(v,i,t,e)$ we have shown that~$R^{**}$ fulfills the three properties from \Cref{def:correct_sub_const}.
\end{proof}
}

Now, we show that using the \emph{Dirty Subset Constraint} we can improve the running time of~$\Oh((\delta\cdot D\cdot s)^{s} \cdot s\cdot  n)$ for \pDTSlong{} shown by \citet{KKSS23} by replacing the~$s$ in the base of the running time by~$\log(s)$.

\begin{theorem}[\appref{thm-dts-improved-algo}]
\label{thm-dts-improved-algo}
  \pDTSlong{} can be solved in $\Oh((\delta\cdot D\cdot \log(s))^{s} \cdot s\cdot  n)$~time.
\end{theorem}
\appendixproof{thm-dts-improved-algo}{
\begin{proof}
Let~$W$ be a witness tree.
Observe that it is sufficient to show that there exists a dirty example~$e\in\dirty(W)$ such that there are at most $\log(s)$~positions in~$W$ such that refining~$W$ at those positions does not violate a \emph{Dirty Subset Constraint}:
since there are at most $\log(s)$ such positions and for each of them we have at most $\delta$~possible dimensions and at most $D$~possible dimensions, we branch into at most $\delta\cdot D\cdot \log(s)$~possibilities.
Hence, it remains to verify this statement.

For this, we construct a compressed witness tree~$W^*$ which we use for a more compact representation of all one-step refinements.
While there exists a leaf~$\ell$ of~$W$ without any dirty example, we do the following:
Let~$p$ be the parent of~$\ell$ in~$W$ and let~$r$ be the (unique) sibling of~$\ell$.
If~$p$ is the root of~$W$, we set~$W'$ to be the subtree of~$W$ rooted in~$r$.
Otherwise, if~$p$ has a parent~$q$ (assume without loss of generality that~$p$ is the right child of~$q$), we construct~$W'$ from~$W$ by replacing the right subtree of~$q$ by the subtree rooted in~$r$ of~$W$, that is, we ''remove''~$p$ and~$\ell$ from~$W$ to obtain~$W'$.
By the definition of the \emph{Dirty Subset Constraint} no refinement is done on the edges~$\{\ell,p\}$ and~$\{p,q\}$ of~$W$ and thus~$W'$ preserves all one-step refinements of~$W$.
Then, we rename~$W'$ to~$W$.

This process stops, if each leaf contains at least one dirty example and then the resulting witness tree is denoted by~$W^*$.
Since in each step all one-step refinements are preserved, we conclude that~$W$ and~$W^*$ represent the same set of one-step refinements.

Let~$\ell$ be a leaf of~$W^*$, which by definition of~$W^*$ contains at least one dirty leaf, of minimal depth~$i$ in~$W^*$.
Since~$d$ is minimal, in each vertex in~$W^*$ with depth~$i'<i$ has to be an inner vertex, representing a cut in~$W$.
These are~$\sum_{j=0}^{i-1}2^j=2^i-1$ many. 
By definition,~$W$ and thus also~$W^*$ has at most $s$~inner vertices.
Hence,~$2^i-1<s$ and thus~$i<\log(s)$.
Consequently, \textsc{BSDT} can be solved in $\Oh((\delta\cdot D\cdot \log(s))^{s} \cdot s\cdot  n)$~time.

Recall that we solve \pDTSlong{} by increasing the value of~$s$ iteratively by one, until the minimal value of~$s$ is found.
Thus, the running time for \pDTSlong{} is~$\Oh(\sum_{i=1}^s (\delta\cdot D\cdot \log(i))^{i} \cdot i\cdot  n)\in \Oh(s\cdot n\cdot \sum_{i=1}^s (\delta\cdot D\cdot \log(s))^{i}) \in \Oh((\delta\cdot D\cdot \log(s))^{s} \cdot s\cdot  n)$ by  the geometric series.
\end{proof}

Unfortunately, the replacement of~$s$ by~$\log(s)$ in the base of the running time is not possible for the more general \textsc{Minimum Tree Ensemble Size (MTES)}.
Roughly speaking, in \textsc{MTES} the majority vote of $\ell$~decision trees is used to classify each example (for a formal definition of \textsc{MTES} we refer to~\citet{KKSS23}).
Hence, an example is dirty if it is misclassified at least $\ell/2+1$~times.
Note that in \textsc{MTES},~$s$ is an upper bound for the number of inner vertices for each of the $\ell$~trees.

For an intuition, consider the case that~$\ell=3$ and that each of the $\ell$~witness trees is a path~$P$ of the same length~$s'=s/2$ using the exact same cuts~$C$ and let~$\pi$ be an ordering of~$C$.
More precisely, one tree of the ensemble performs the cuts in order~$\pi$, the second tree performs the cuts in~$\pi$ with a cyclic shift of~$s'/3$, and the third tree performs the cuts in~$\pi$ with a cyclic shift of~$2s'/3$.
Now, observe that for each dirty example~$e$ we need to branch on at least $s'/3\in \Oh(s)$~positions to decrease the number of misclassifications of~$e$ by at least one.

}

\section{Subset caching}
\label{sec:subset_caching}
\appendixsection{sec:subset_caching}

Our final improvement is inspired by the caching of subproblems used in MurTree~\citep{DBLP:journals/jmlr/DemirovicLHCBLR22} to solve the related problem of minimizing the number of misclassifications under a size and depth constraint.
They iterate over all possible cuts for the root and then calculate optimal solutions for the left and right subtree recursively.
For this approach caching of subproblems is very natural.

Since \OS{} can modify any part of the current witness tree, \OS{} is not directly amenable to caching.
Thus, we modify the caching of subproblems as follows:
We use a \emph{set-trie data structure}~\citep{DBLP:conf/IEEEares/Savnik13} to save lower bounds for specific subsets of the examples. 
These lower bounds tell us how many inner vertices we need at least to correctly classify that subset of the examples. 
In any search-tree node~$N$ with~$W\coloneqq \tree(N)$, we then look at each leaf~$\ell\in V(W)$ and the set of examples~$L\coloneqq E[W,\ell]$ assigned to~$\ell$ and check if a subset~$S$ of~$L$ is present in our data structure. 
If this is the case, then the lower bound~$z$ associated with~$S$ is also a lower bound for~$L$.
Since all examples in~$L$ are assigned to the same leaf, this is also a valid lower bound for~$W$.
If~$z$ is bigger than the remaining size budget~$s'$, we know that it is not possible to correctly classify the data set by applying at most~$s'$ one-step refinements to~$W$.
Consequently, the algorithm can return~$\bot$. 
Otherwise, if~$z$ is not big enough, we keep checking our data structure for another subset of~$L$ until we either find a sufficiently large lower bound or until no more subsets of~$L$ are left.

\toappendix{
We use the \emph{set-trie data structure}~\citep{DBLP:conf/IEEEares/Savnik13} to store the lower bounds for the subsets of examples.
%Of course for all of this to work, we need a data structure that lets us save lower bounds for subsets of the examples and lets us quickly check out all subsets of a specific set of examples. 
%For this, we use the set-trie data structure~\cite{DBLP:conf/IEEEares/Savnik13}. 
A set-trie~$\mathcal{T}$ saves the subsets in a tree structure where each vertex represents a single example. 
We add a subset~$S$ of examples to~$\mathcal{T}$ by first sorting~$S$ in ascending order~$\tau$. 
We then start at the first example~$e_1$ in~$\tau$ and check if the root of~$\mathcal{T}$ has a child~$c$ that represents~$e_1$. 
If not, we add such a child. 
Next, we check if~$c$ has a child that represents the second example~$e_2$ and add a new child if that is not the case. 
We continue doing this until we have created a path that represents~$S$. 
We then mark the last vertex in this path as an end vertex and save the lower bound of the subset in this end vertex.
}

Since our algorithm does not naturally calculate lower bounds for subsets of examples, to populate the set trie~$\mathcal{T}$ we do the following: 
%Instead, we do the following: 
In a search-tree node~$N$, we look at the example set~$L\coloneqq E[\tree(N),\ell]$ of any leaf~$\ell$ with~$|L|\leq \min(|E| / 4,30)$. 
Let~$s'$ be the remaining size budget of the current witness tree. 
We run a separate instance of our algorithm that checks whether the set~$L$ can be correctly classified with a tree of size at most~$s'$. 
If not, then we can add~$L$ to~$\mathcal{T}$ with the lower bound~$s'+1$. 
Of course we could then check if~$L$ can be correctly classified with a tree of size at most~$s'+1$ to improve the lower bound, but we decided to only check size~$s'$ because that is sufficient to show that the current tree cannot correctly classify the data. 
If we ever need a better lower bound for~$L$ at a later point, we can still calculate it then. 
A limit of the size of~$L$ is necessary since otherwise one  solves the original problem; preliminary experiments showed that~$\min(|E| / 4,30)$ is a good limit.

%Lastly, we want to mention that the set-trie is especially useful when solving \textsc{MSDT} instead of only~\textsc{BSDT} since we do not have to throw the set-trie away after the algorithm has checked a single~$s$. We can keep reusing the set-trie until we find the optimal~$s$.

A set-trie is especially useful when solving \textsc{MSDT} since we can reuse the set-trie for each instance of
\textsc{BSDT} until we find the optimal~$s$.
For all solved instances, each trie had at most $80\ 000$~vertices, resulting in a maximum space consumption of $3$~GB.

\section{Experimental evaluation}
\label{sec:evaluation}
\appendixsection{sec:evaluation}

\paragraph{Experimental setup.}

For our experiments, we used~$35$ data sets that were also used in the experimental evaluation of the state-of-the-art SAT-based \textsc{MSDT} solver~\citep{DBLP:conf/sat/JanotaM20}.
The data sets are part of the Penn Machine Learning Benchmarks~\citep{romano2021pmlb}; \Cref{tab:overview_datasets_reduced} (in the appendix) gives an overview.

%In the following sections we present several improvements to Algorithm~\ref{alg:base}. In order to show the effect of these improvements, we show the results of experiments where we only use the improvements we have shown up to that point and the new improvement. We then compare those results to the results of the previous section. After we have presented all of the improvements, we compare the final version of our algorithm with the algorithms presented by Janota et al.~\cite{DBLP:conf/sat/JanotaM20} that also solve \textsc{MSDT}. In this section we explain the details of how we performed those experiments.

To meet the requirements of \textsc{MSDT} inputs, we transformed the data sets as follows (similar to \citet{DBLP:conf/sat/JanotaM20}). We replaced each categorical dimension by a set of new binary dimensions indicating whether an example is in the category.
Second, we converted each instance into a binary classification problem.
For this, we colored all examples of the largest class red, and all remaining examples blue.
Finally, if two examples had the same value in all dimensions but are colored differently, we removed one of them arbitrarily.

Following \citet{DBLP:conf/sat/JanotaM20}, we randomly sampled multiple subsets of the examples from each data set. 
Specifically, for each data set we chose~$10$ random subsets with~$20\%$ of the examples and~$10$ random subsets with~$50\%$ of the examples. 
This means we ran our experiments on a total of~$700$ instances. 
For each instance we used a time limit of~$60$ minutes.

Our experiments were performed on servers with~$24$ GB RAM and two Intel(R) Xeon(R) E5540 CPUs with~$2.53$ GHz,~$4$ cores,~$8$ threads, running Java openjdk 11.0.19. Each individual experiment was allowed to use up to~$12$ GB RAM. All algorithms we tested run on a single core. However, we reserved two additional cores for each experiment to avoid side effects of other threads like the Java garbage collector. We implemented our algorithm in Kotlin and the code is available in the supplementary material. To solve the LP-relaxation of the PairLB we used Gurobi 10.0.3 \citep{gurobi}.

\begin{figure}[t]
	\centering
	\includegraphics[width=\linewidth]{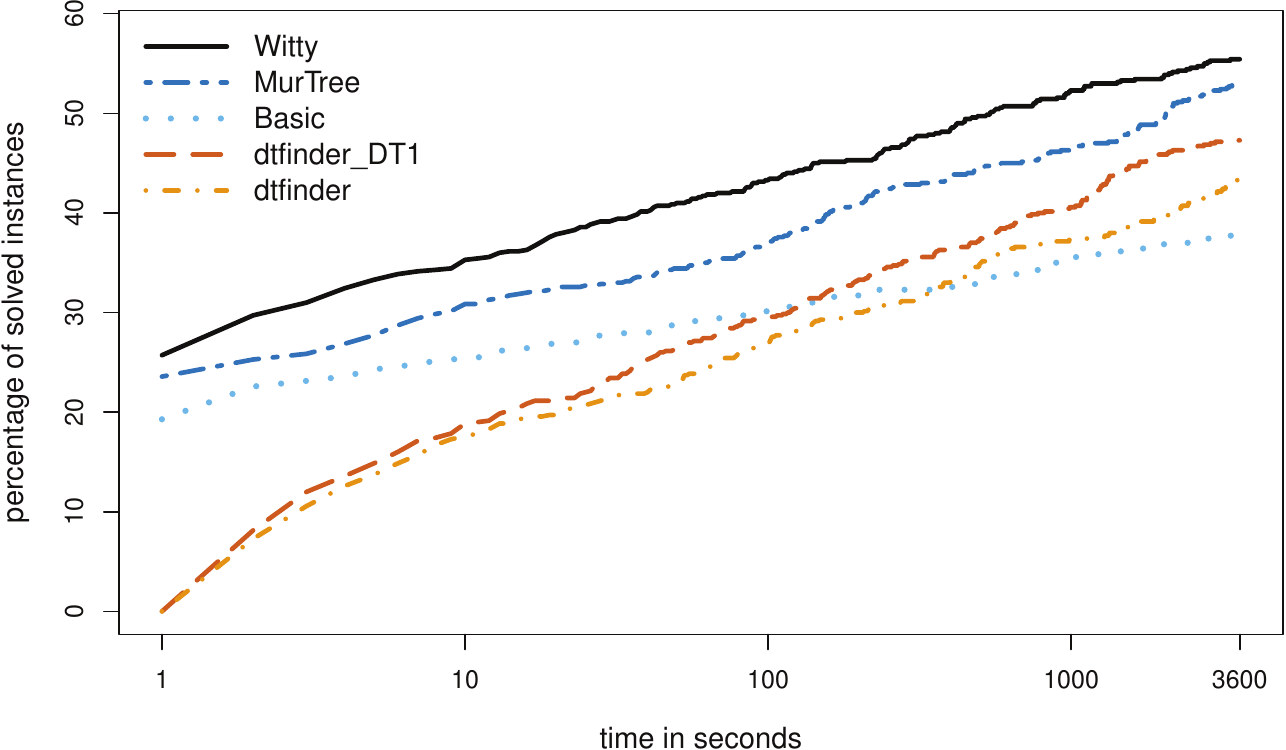}
	\captionof{figure}{Comparison of different algorithms for \textsc{MSDT}. For each time~$t$ it is shown how many instances were solved by each algorithm in less than $t$~seconds.}
	\label{fig-variants}
\end{figure}

\begin{figure}[t]
	\centering
	\includegraphics[width=.9\linewidth]{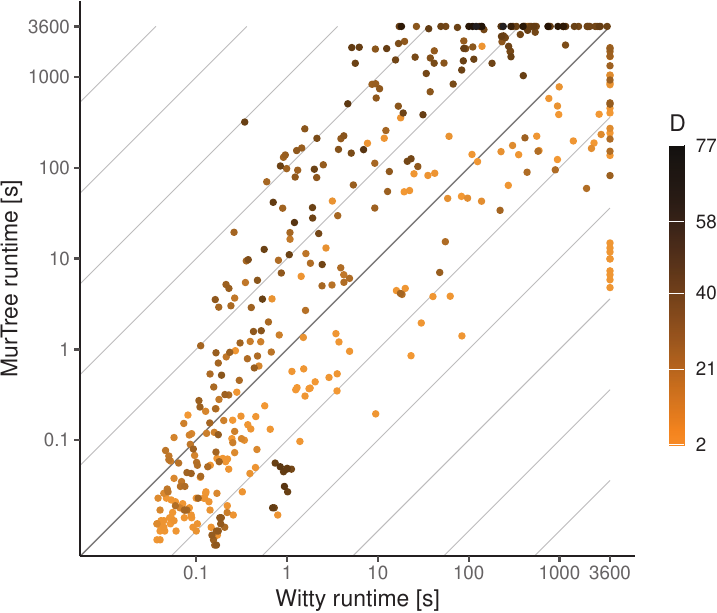}
	\captionof{figure}{Comparison of the running times of \OS{} and \texttt{MurTree} for each instance with the color representing the largest domain size~$D$.}
	\label{fig-time-D}
\end{figure}

\paragraph{Comparison with the basic version of \OS.}

\Cref{fig-variants} shows the performance of \OS{} and \texttt{Basic}. \OS{} is the final version of our algorithm using all improvements while \texttt{Basic} is just Algorithm~\ref{alg:base} together with the dirty example priority and the data reduction rules from Section~\ref{sec:base_algo}.

\texttt{Basic} solved~$264$ out of the~$700$ instances. In comparison the naive version of our algorithm without any improvements only solved~$222$ instances and \texttt{Basic} is~$26$ times (median~$9$ times) faster on instances solved by both. This shows that the dirty example priority and data reduction rules from \Cref{sec:base_algo} already lead to a large speed-up.

\OS{} solved~$388$ instances and is roughly~$65$ times (median~$23$ times) faster than \texttt{Basic} on instances that were solved by both, showing that the combination of all of our improvements yields a vast speed-up.

A more detailed comparison of the different solver configurations can be found in the appendix.

\todo[inline]{maybe briefly mention how many instances were solved by LB and/or SubConst

F: I wouldnt do this}

%\looseness=-1
%The performance of four variants of our algorithm is shown in \Cref{fig-variants}; \texttt{Basic} only uses the dirty example priority and the data reduction rules from Section \ref{sec:base_algo}, \texttt{LB} additionally uses the lower bounds from Section \ref{sec:lower_bounds}, \texttt{SubConst} also uses the Subset Constraints from Section \ref{sec:subset_constraints} and \OS{} uses all improvements.

%\OS{} solved~$388$ out of the~$700$ instances. It solved $6$~instances more than \texttt{SubConst} which solved $52$~instances more than \texttt{LB} which solved $66$~instances more than \texttt{Basic}.
%\OS{} is roughly $1.6$ times (median~$1.1$ times) faster than \texttt{SubConst} which is roughly $8$ times (median $3.4$ times) faster than \texttt{LB} which is roughly $22$ times (median $8$ times) faster than \texttt{Basic}.
%Overall, \OS{} is roughly $65$ times (median~$23$~times) faster than \texttt{Basic}, showing that the combination of all of our improvements yields a vast speed-up. Additionally, \OS{} is $324$ times (median $84$ times) faster than the naive version of our algorithm without any improvements showing that the dirty example priority and data reduction rules from \Cref{sec:base_algo} also lead to a large speed-up.

\paragraph{Comparison with the state of the art.}

We compare \OS{} against the state-of-the-art SAT-based algorithms~\citep{DBLP:conf/ijcai/NarodytskaIPM18,DBLP:conf/sat/JanotaM20}. 
We use the names \texttt{dtfinder\_DT1} and \texttt{dtfinder}, respectively, to refer to these algorithms. 
We use the improved version of the encoding by \citet{DBLP:conf/ijcai/NarodytskaIPM18} that was presented by \citet{DBLP:conf/sat/JanotaM20}.
We also compare \OS{} against the state-of-the-art dynamic programming based solver, \texttt{MurTree}~\citep{DBLP:journals/jmlr/DemirovicLHCBLR22}. Since these algorithms only support binary dimensions, we transform the data sets following their approach by introducing a binary dimension for each cut indicating wether an example is on the left or right side of the cut. We still used the non-binarized data sets for \OS{}.

\Cref{fig-variants} shows this comparison.
\OS{} solved~$57$~instances more than \texttt{dtfinder\_DT1} which solved~$26$~instances more than \texttt{dtfinder}.
\OS{} is roughly~$61$ times (median~$25$ times) faster than \texttt{dtfinder\_DT1}.
Furthermore, there is only one instance that can be solved by \texttt{dtfinder} or \texttt{dtfinder\_DT1} but not by \OS.

\texttt{MurTree} on the other hand solved~$371$ instances. However, only~$341$ of these instances were solved by both \OS{} and \texttt{MurTree}. 
On these instances, \OS{} achieved a mean~$32$-fold (median~$7$-fold) speedup over \texttt{MurTree}.
To compute the speed-up, we ignored instances that were solved in less than one second by both algorithms. 
One notable property of the instances that \OS{} solved faster than \texttt{MurTree} is that the largest domain size~$D$ tends to be a lot bigger than in the instances that \texttt{MurTree} could solve faster than \OS{}. This can be seen in \Cref{fig-time-D}.
We assume that the Threshold Subset Constraints from Section~\ref{sec:subset_constraints} are the reason for this: these Subset Constraints are particularly effective on dimensions with many thresholds since they allow the algorithm to skip several thresholds within such a dimension.

We also looked at the values of the instance-specific parameters~$n,d,\delta,c$, and~$s$ of these instances (see \Cref{fig-time-n,fig-time-d,fig-time-delta,fig-time-c,fig-time-s} in the appendix).
There are two notable observations. First, \OS{} also tends to perform better than \texttt{MurTree} on instances with a large number of cuts~$c$. This is most likely related to the above observation that these instances also tend to have a large value for~$D$.
Second, \texttt{MurTree} can solve instances with a larger optimal tree size~$s$. \texttt{MurTree} solved instances with~$s \leq 20$ while \OS{} could only solve instances with~$s \leq 16$. This suggests that \texttt{MurTree} scales better with the size of the optimal tree as long as the largest domain size~$D$ is not too large. In fact, all instances with~$s > 16$ that \texttt{MurTree} solved have a value of~$D = 2$.

\paragraph{Comparison with heuristics.}

We evaluated the size, classification quality on the test data, and the balance of the trees computed by \OS{} against three heuristics: CART~\citep{DBLP:books/wa/BreimanFOS84}, Weka~\citep{WittenFH16}, and YaDT~\citep{Ruggieri04,Ruggieri19}.
CART and Weka also compute perfect trees.
Our evaluation (see \Cref{tab:overview_datasets_tree_size_comparison,tab:overview_datasets_test_classification_comparison,tab:overview_datasets_tree_balance_comparison} in the appendix) shows that the trees of YaDT are much smaller than the trees of \OS{} which are a bit smaller than the ones of CART and Weka.
Also, all computed trees are very balanced and achieve similar classification quality.

%%%%%%%%%%%%%%%%%%%%%%%%%%%%%%%%%%%%%%%%%%%%%%%%%%%%%%%%%%%%%%%%%%%%%%%%%%%%%%%%%%%%%%%%%
%%% APPENDIX %%%%%%%%%%%%%%%%%%%%%%%%%%%%%%%%%%%%%%%%%%%%%%%%%%%%%%%%%%%%%%%%%%%%%%%%%%%%
%%%%%%%%%%%%%%%%%%%%%%%%%%%%%%%%%%%%%%%%%%%%%%%%%%%%%%%%%%%%%%%%%%%%%%%%%%%%%%%%%%%%%%%%%
\toappendix{
\subsection{Evaluation of different variants of our algorithm \OS{}}

\begin{figure}
\centering
\begin{minipage}{.45\textwidth}
	\centering
	\includegraphics[width=\linewidth]{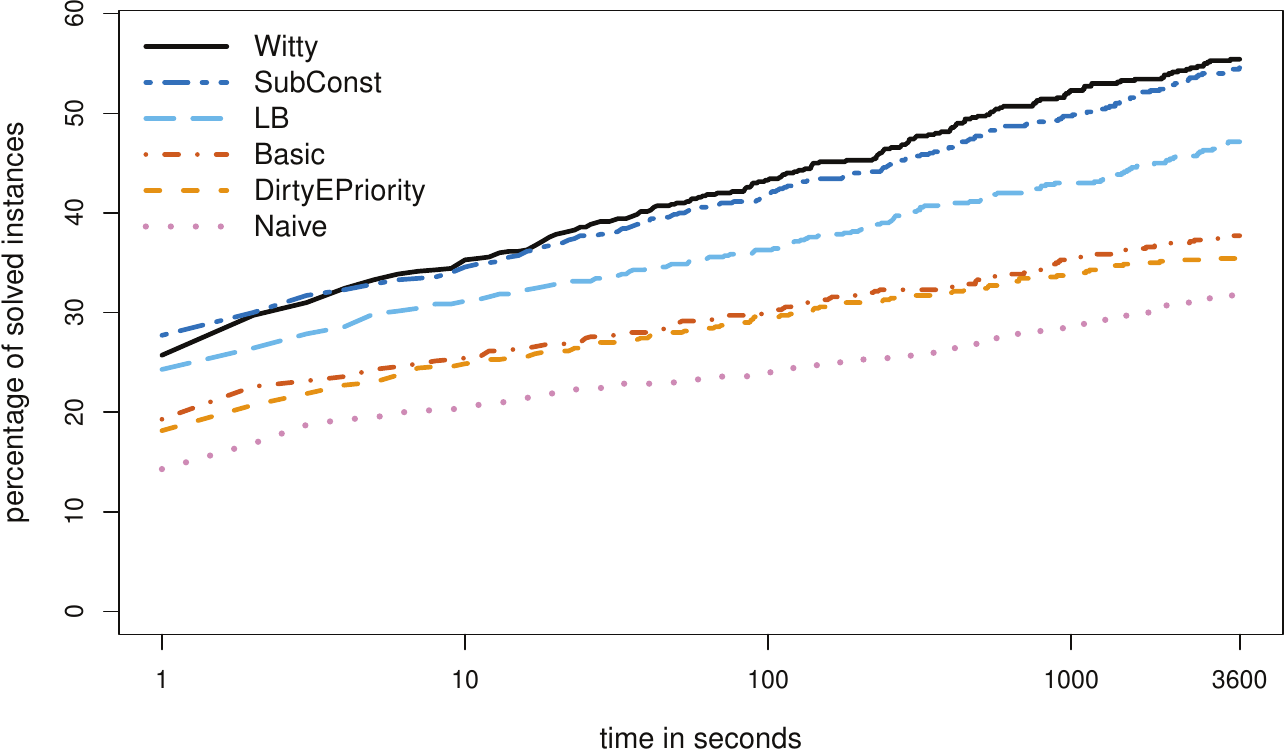}
	\captionof{figure}{Effect of the different improvements on the running time.}
	\label{fig-dirty-ex-prio-use-rr}
\end{minipage}%
\hspace{0.3cm}
\begin{minipage}{.45\textwidth}
	\centering
	\includegraphics[width=\linewidth]{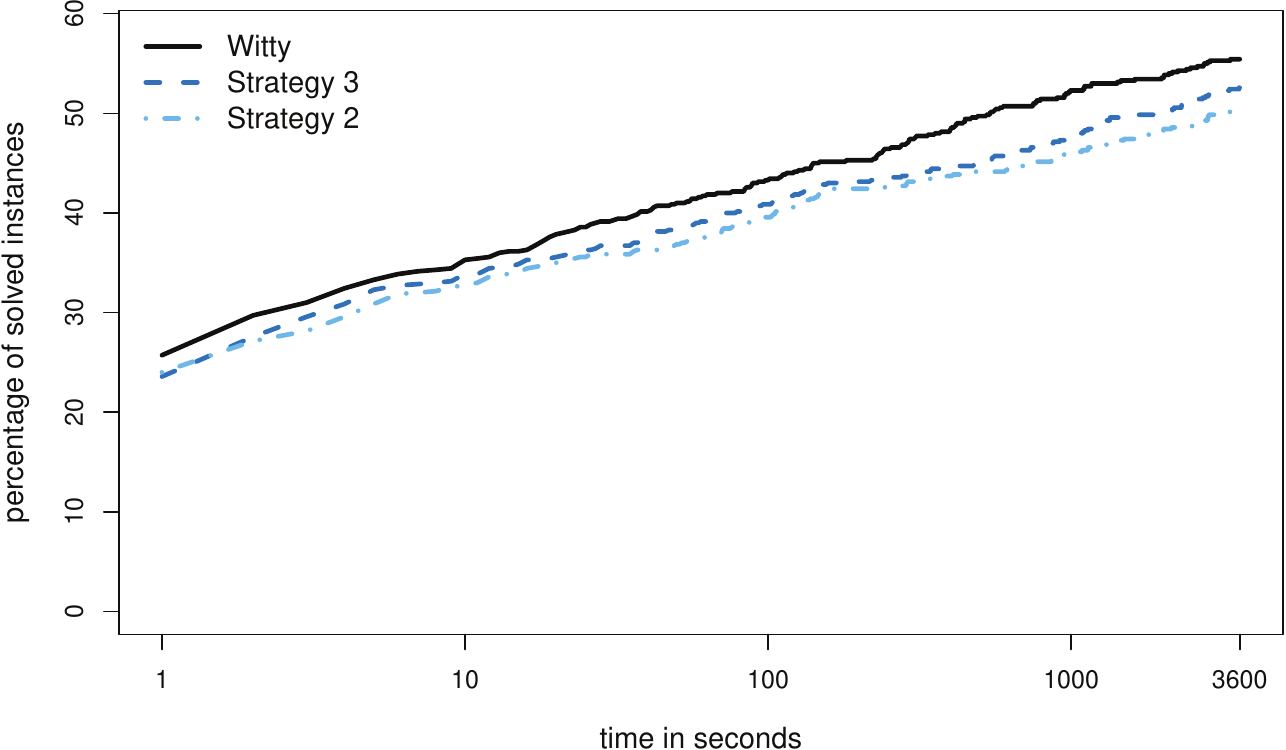}
	\caption{Comparison of the three strategies for solving \textsc{MSDT}.}
	\label{pdf:comparison_strategies}
\end{minipage}
\end{figure}

Here we compare the performance of six variants of our algorithm. This comparison is shown in \Cref{fig-dirty-ex-prio-use-rr}. \texttt{Naive} does not use any of the improvements discussed in this paper. \texttt{DirtyEPriority} only uses the dirty example priority from \Cref{sec:base_algo}. \texttt{Basic} uses both the dirty example priority and the reduction rules from \Cref{sec:base_algo}. \texttt{LB} additionally uses the lower bounds from Section \ref{sec:lower_bounds}, \texttt{SubConst} also uses the Subset Constraints from Section \ref{sec:subset_constraints} and \OS{} uses all improvements.

\paragraph{Effect of using the dirty example priority.}
\texttt{DirtyEPriority} solved $26$~instances more than \texttt{Naive}.
Only $13$~instances were solved more than a second faster by \texttt{Naive} than \texttt{DirtyEPriority}.
\texttt{DirtyEPriority} also substantially decreases the size of the search trees: 
on average, the search trees of \texttt{DirtyEPriority} had $20$~times less nodes than the search trees of the same instances solved by \texttt{Naive}.
Thus, using the dirty example priority is clearly beneficial.

\paragraph{Effect of using the reduction rules.}
 
First, we explain the order in which we use the reduction rules. 
This is important because applying the same rules to the same data set in two different orders may lead to different results (for example: Dimension Reduction Rule and Dimension Merge Rule).

Initially, we remove as many cuts as possible: 
first, we apply the Dimension Reduction Rule exhaustively,  second we apply the Equivalent Cuts Rule exhaustively. 
Afterwards, we apply the Dimension Merge Rule exhaustively to remove as many dimensions as possible. 
Finally, we apply the Remove Duplicate Examples Rule and the Remove Dimension Rule exhaustively to clean up the data. 

For a general overview on the effectiveness on our data rules, we refer to \Cref{tab:overview_datasets_reduced}.
In most cases our reduction rules reduced the number of cuts; for some instances even up to one third of all cuts.

As can be seen in \Cref{fig-dirty-ex-prio-use-rr}, the use of the reduction rules resulted in $16$~more solved instances.
There were, however, two instances that were only solved by \texttt{DirtyEPriority}. 
Furthermore, $12$~instances were solved more than a second faster by \texttt{DirtyEPriority} than by \texttt{Basic}.
Hence, overall \texttt{Basic} performed slightly better than \texttt{DirtyEPriority}.
Moreover, the main advantage of these reduction rules is that they are a protection against data sets with a high level of redundancy.

\paragraph{Effect of using the lower bounds.}
As discussed in \Cref{sec:lower_bounds}, we calculate the lower bounds after Line~$5$ in each call of \texttt{Refine} in \Cref{alg:base} and return~$\bot$ if one of them is bigger than the remaining size budget. 
However, preliminary experiments showed that the calculation of the PairLB takes too long compared to how effective it is. 
Because of this we only calculate the PairLB once to obtain an initial lower bound for the solution of the \textsc{MSDT} instance.

\texttt{LB} solved $66$~instances more than \texttt{Basic} and on instances solved by both \texttt{LB} was on average~$22$ times (median~$8$ times) faster than \texttt{Basic}.
All instances solved by \texttt{Basic} were solved at most one second slower by \texttt{LB}.
Furthermore, \texttt{LB} decreases the number of search-tree nodes drastically:  
on average, the search trees of \texttt{LB} had~$64$ times less nodes than the search trees of \texttt{Basic} on the same instances. 
Thus, the two lower bounds are a big improvement for the algorithm in every instance.

\paragraph{Effect of using the Subset Constraints.}
We check in \Cref{line-apply-refine} of \Cref{alg:base} whether any subset constraint is violated, right after the one-step refinement has been performed.
Overall, the Subset Constraints were one of the biggest improvements for the number of solvable instances. \texttt{SubConst} solved $52$~instances more than \texttt{LB}. On instances solved by both \texttt{SubConst} was roughly $8$ times (median $3.4$ times) faster than \texttt{LB}.
All instances solved by \texttt{LB} were also solved by \texttt{SubConst}.
Furthermore, all of these instances were solved faster by \texttt{SubConst} than by \texttt{LB}.
\texttt{SubConst} also reduced the number of search-tree nodes substantially: 
on average, the search trees of \texttt{SubConst} were~$13$ times smaller than the search trees of \texttt{LB} on the same instances. 
Thus, the two subset constraints are also generally a big improvement for the algorithm.

\paragraph{Effect of using the subset caching.}

\OS{} solved six instances more than \texttt{SubConst} and was roughly~$1.6$ times (median~$1.1$ times) faster than \texttt{SubConst} on instances solved by both.
There were only three instances which were solved by \texttt{SubConst} but not by \OS.
Since \OS{} solved more instances than \texttt{SubConst} and yielded a speed-up of factor~$1.6$ (median~$1.1$), subset caching is a small improvement for the algorithm.
However, on average, the search trees of \OS{} had~$1.15$ times more nodes than the search trees of \texttt{SubConst} on the same instances.
This is clearly due to the additional instances of \textsc{BSDT} that we solve to populate the set-trie and it shows why it is important for us to limit which example subsets we choose to check and potentially add to the set-trie (recall that we set the limit to~$\min(|E| / 4,30)$).

\paragraph{Comparison of the three different strategies to solve \textsc{MSDT}.}

Recall that our algorithm only solves the decision problem \textsc{BSDT}. 
But since we are interested in the optimization problem \textsc{MSDT} we presented a strategy that allows us to solve \textsc{MSDT} by solving multiple instances of \textsc{BSDT} using our algorithm.
In the strategy we described in the main part (\OS), we linearly increase the maximum size of the tree starting from a lower bound until our algorithm no longer returns~$\bot$. 
However, this is not the only way to find the optimal value of~$s$.

As a second strategy we can linearly decrease the maximum size of the tree starting from an upper bound until our algorithm returns~$\bot$. 
As a third strategy we can perform a binary search between a lower bound and an upper bound until the optimal tree size is found.

For the initial lower bound we use the PairLB from \Cref{sec:pairLB}. 
For the initial upper bound we use the scikit-learn implementation~\citep{scikit-learn} of the CART algorithm originally proposed by \citet{DBLP:books/wa/BreimanFOS84}. 
The SAT-based solver by~\citet{DBLP:conf/sat/JanotaM20} also uses this implementation to calculate an upper bound.

\Cref{pdf:comparison_strategies} shows the comparison of the three different strategies.
\OS{} solved $20$~instances more than Strategy~$3$ which solved $16$~instances more than Strategy~$2$.
Overall, \OS{} is roughly $4.5$ times (median~$1.6$ times) faster than Strategy~$2$, the slowest of the three strategies.
Furthermore, all instances that were solved by either one of the other $2$~strategies were also solved by \OS{}.
If we compare the running time of the three strategies on instances that can be solved by all of them we can see that \OS{} is only very rarely substantially slower than the other two: 
there are only~$14$~instances where \OS{} is more than~$1$ second slower than the fastest strategy.
This shows that \OS{} is generally the best choice for our algorithm.

The reason why the three strategies perform similarly is most likely because there are always at least two instances of \textsc{BSDT} that need to be solved by all strategies: 
If~$s_{\min}$ is the solution for an instance of \textsc{MSDT}, then all strategies need to solve at least the \textsc{BSDT} instance with~$s = s_{\min}$ and the instance with~$s = s_{\min} - 1$.
Since these two instances are closest to the optimal value $s_{\min}$ they take the longest time to solve and are therefore the main deciding factor in whether a strategy can solve an instance of \textsc{MSDT} or not.

\subsection{Dependence of \OS{} and \texttt{MurTree} on different input parameters}

\begin{figure}
\centering
\begin{minipage}{.4\textwidth}
  \centering
  \includegraphics[width=.9\linewidth]{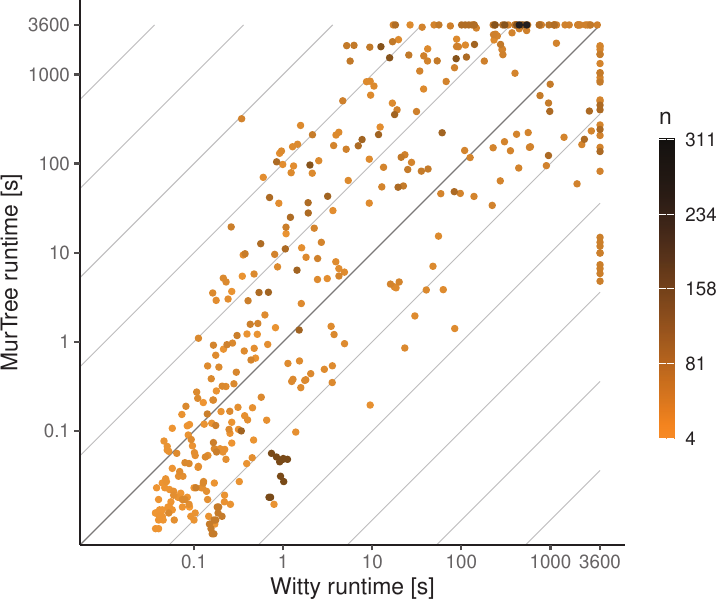}
  \captionof{figure}{Relation of number~$n$ of examples to the time required to solve the instance.}
  \label{fig-time-n}
\end{minipage}%
\hspace{0.3cm}
\begin{minipage}{.4\textwidth}
  \centering
  \includegraphics[width=.9\linewidth]{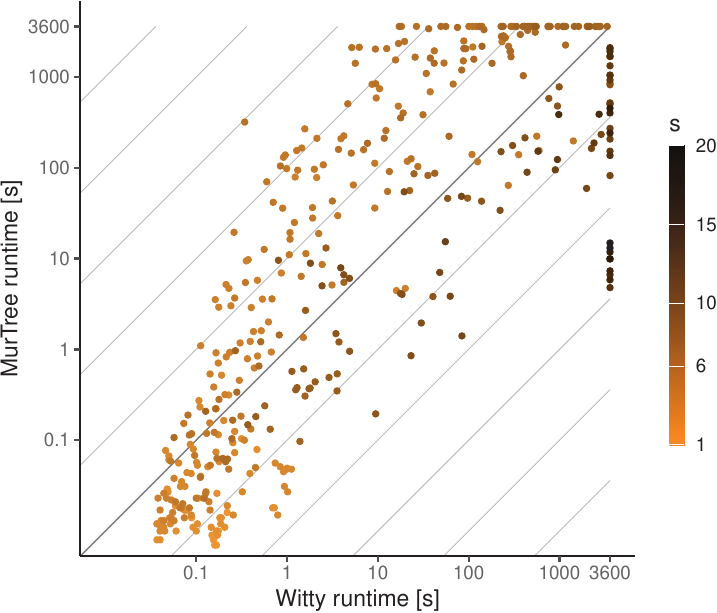}
  \captionof{figure}{Relation of smallest tree size~$s$ to the time required to solve the instance.}
  \label{fig-time-s}
\end{minipage}
\end{figure}

\begin{figure}
\centering
\begin{minipage}{.4\textwidth}
  \centering
  \includegraphics[width=.9\linewidth]{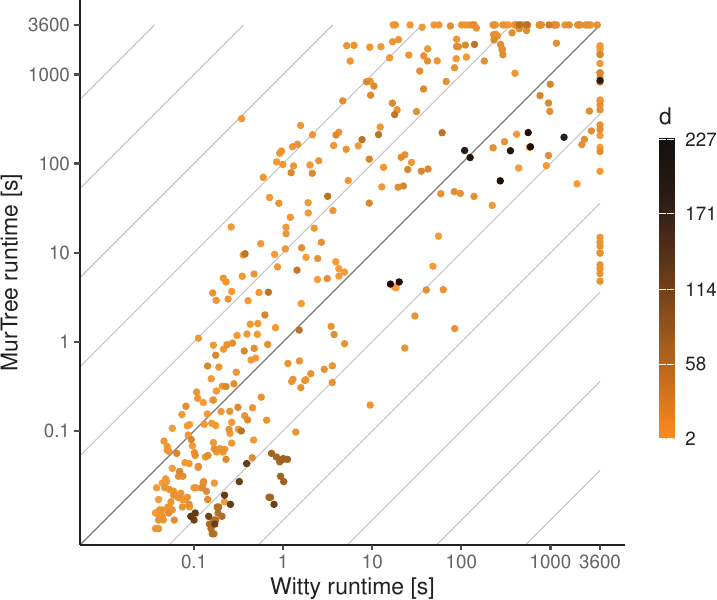}
  \captionof{figure}{Relation of number~$d$ of dimensions to the time required to solve the instance.\\}
  \label{fig-time-d}
\end{minipage}%
\hspace{0.3cm}
\begin{minipage}{.4\textwidth}
  \centering
  \includegraphics[width=.9\linewidth]{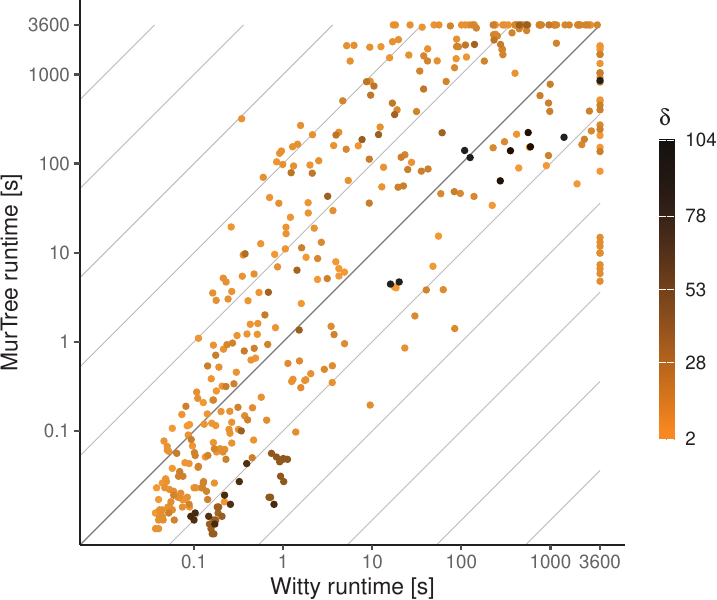}
  \captionof{figure}{Relation of number~$\delta$ of maximal dimensions in which two different labeled examples differ to the time required to solve the instance.}
  \label{fig-time-delta}
\end{minipage}
\end{figure}

\begin{figure}
	\centering
  \includegraphics[width=.36\textwidth]{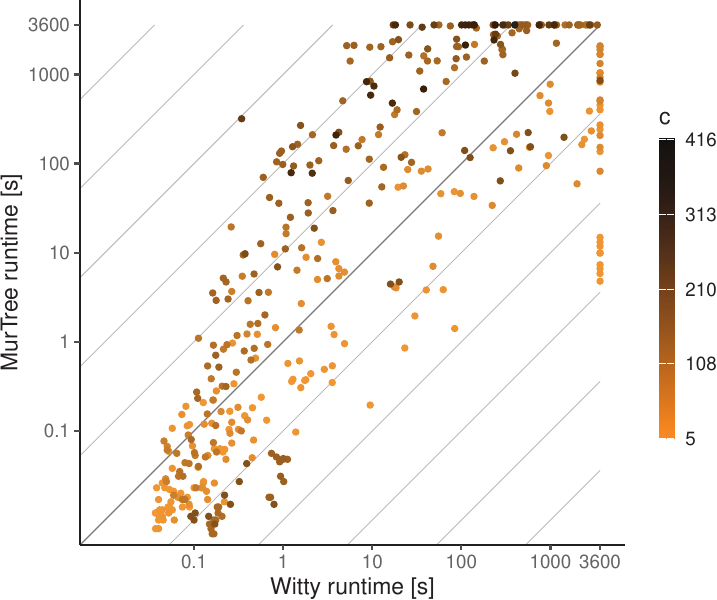}
  \captionof{figure}{Relation of number~$c$ of cuts to the time required to solve the instance.}
  \label{fig-time-c}
\end{figure}

We compared the running times of our algorithm \OS{} and the state-of-the-art \texttt{MurTree} against the parameters number~$n$ of examples (\Cref{fig-time-n}), number~$d$ of dimensions (\Cref{fig-time-d}), largest domain size~$D$ (\Cref{fig-time-D})\footnote{Recall that for \texttt{MurTree} it was necessary to transform the data set~$\mathcal{D}$ into an equivalent data set~$\mathcal{D}'$ where each dimension is binary (hence~$D=2$). For \Cref{fig-time-D} we used the original data set~$\mathcal{D}$.}, number~$\delta$ of maximal dimensions in which two different labeled examples differ (\Cref{fig-time-delta}), number~$c$ of cuts (\Cref{fig-time-c}), and size~$s$ of a solution (\Cref{fig-time-s}).
Note that in all six figures the light gray lines correspond to a difference of factor 10 in the running time of both algorithms.
Furthermore, observe that the location of all points in the six figures is identical, only their coloring is different; depending on the parameter~$n,d,D,\delta,c,$ or~$s$, respectively.

Overall we can differentiate between "easy" instances that can be solved in at most one second by both algorithms and "hard" instances where at least one algorithm takes longer than one second.
The figures show that \texttt{MurTree} tends to be faster than \OS{} on "easy" instances. However, on "hard" instances \OS{} performs a bit better than \texttt{MurTree} suggesting that \OS{} scales a better than \texttt{MurTree}.
Additionally, \OS{} has more instances with a speedup between~$10$ and~$100$ over \texttt{MurTree} than \texttt{MurTree} has over \OS{}.

\Cref{fig-time-n} does not show a clear trend. However, there do seem to be a few more instances with large~$n$ where \OS{} performs better than \texttt{MurTree} but this is not enough to make a clear judgement on which algorithm scales better with~$n$.

As previously mentioned in Section \ref{sec:evaluation}, \Cref{fig-time-s} suggests that \texttt{MurTree} scales slightly better with~$s$ than \OS{}. However, if we also look at \Cref{fig-time-D} we can see that the instances with a large value for~$s$ where \texttt{MurTree} is better than \OS{} also have very small values of~$D$. Together with \Cref{fig-time-c}, this suggests that \OS{} can deal with a larger number of cuts as long as these cuts are part of a few larger dimensions. But on instances with small values of~$D$, \texttt{MurTree} performs better and can even solve instances with a larger value for~$s$.

Finally, \Cref{fig-time-d,fig-time-delta}, illustrating the dependence on~$d$ and~$\delta$, do not show a clear pattern. There are a few instances where these values are large that are solved slightly faster by \texttt{MurTree} than by \OS{} but they are not enough to really form a conclusion. Especially since most of these instances can be solved in less than a second by both algorithms.

\begin{table}[t!]
	\centering
	%\scalebox{0.75}{
	\small
	\begin{tabular}{lrrrrrrrr}
\toprule
	\multirow{2}{*}{Instance name} & \multicolumn{2}{c}{\OS}  & \multicolumn{2}{c}{CART} & \multicolumn{2}{c}{Weka} & \multicolumn{2}{c}{YaDT} \\ \cmidrule(r){2-9}
 &  $20\%$ & $50\%$ & $20\%$ & $50\%$ & $20\%$ & $50\%$ & $20\%$ & $50\%$ \\ \midrule
		postoperative-patient-data & $3.80$ & $9.90$ & $4.40$ & $14.80$ & $4.80$ & $15.40$ & $1.20$ & $1.90$ \\
		hayes-roth & $3.80$ & $8.70$ & $4.60$ & $11.70$ & $4.40$ & $11.60$ & $1.70$ & $4.10$ \\
		lupus & $4.90$ & $11.10$ & $5.30$ & $13.10$ & $5.30$ & $12.80$ & $1.20$ & $1.20$ \\
		appendicitis & $2.40$ & $5.20$ & $2.90$ & $6.80$ & $2.90$ & $7.40$ & $1.30$ & $1.90$ \\
		molecular\_biology\_promoters & $2.50$ & - & $2.70$ & - & $3.00$ & - & $2.40$ & - \\
		tae & $5.60$ & - & $6.90$ & - & $7.00$ & - & $0.90$ & - \\
		cloud & $3.50$ & - & $4.00$ & - & $4.00$ & - & $2.40$ & - \\
		cleveland-nominal & $7.60$ & - & $9.00$ & - & $9.30$ & - & $3.60$ & - \\
		lymphography & $4.20$ & - & $5.30$ & - & $5.20$ & - & $2.80$ & - \\
		hepatitis & $3.50$ & - & $4.20$ & - & $4.40$ & - & $2.40$ & - \\
		glass2 & $4.80$ & - & $7.00$ & - & $5.90$ & - & $3.90$ & - \\
		backache & $3.80$ & - & $4.60$ & - & $4.50$ & - & $2.20$ & - \\
		auto & $4.80$ & - & $5.80$ & - & $6.80$ & - & $5.40$ & - \\
		glass & $5.40$ & - & $7.80$ & - & $7.80$ & - & $3.80$ & - \\
		biomed & $4.00$ & $6.30$ & $5.50$ & $12.20$ & $5.20$ & $11.20$ & $2.60$ & $4.70$ \\
		new-thyroid & $3.10$ & $5.30$ & $3.30$ & $6.60$ & $3.90$ & $7.10$ & $2.20$ & $5.00$ \\
		spect & $6.60$ & - & $8.10$ & - & $8.60$ & - & $1.10$ & - \\
		breast-cancer & - & - & - & - & - & - & - & - \\
		heart-statlog & $6.80$ & - & $9.40$ & - & $10.10$ & - & $4.60$ & - \\
		haberman & - & - & - & - & - & - & - & - \\
		heart-h & - & - & - & - & - & - & - & - \\
		hungarian & $7.20$ & - & $9.50$ & - & $9.90$ & - & $3.50$ & - \\
		cleve & - & - & - & - & - & - & - & - \\
		heart-c & - & - & - & - & - & - & - & - \\
		cleveland & - & - & - & - & - & - & - & - \\
		ecoli & $3.70$ & - & $5.20$ & - & $5.20$ & - & $3.20$ & - \\
		schizo & - & - & - & - & - & - & - & - \\
		bupa & - & - & - & - & - & - & - & - \\
		colic & - & - & - & - & - & - & - & - \\
		dermatology & $1.50$ & $2.80$ & $1.50$ & $3.20$ & $2.00$ & $4.40$ & $1.10$ & $1.80$ \\
		cars & $6.20$ & - & $8.90$ & - & $8.50$ & - & $5.10$ & - \\
		soybean & $6.10$ & - & $8.50$ & - & $9.20$ & - & $5.10$ & - \\
		australian & - & - & - & - & - & - & - & - \\
		diabetes & - & - & - & - & - & - & - & - \\
		contraceptive & - & - & - & - & - & - & - & - \\ 
		\bottomrule
	\end{tabular}
	%}
	\caption{Overview of the tree sizes of the solutions. 
	Recall that for each instance the training data set consists of~$20\%$ or~$50\%$ of the examples and for both values we sampled 10~times.
	If \OS{} solved at least~$5$ of the $10$~sampled instances within the time limit, the corresponding entry is the average solution size.
	Each of the three heuristics was only used for instances which were solved by \OS{}.}
	\label{tab:overview_datasets_tree_size_comparison}
\end{table}

\begin{table}[t!]
	\centering
	%\scalebox{0.75}{
	\small
	\begin{tabular}{lrrrrrrrr}
\toprule
	\multirow{2}{*}{Instance name} & \multicolumn{2}{c}{\OS}  & \multicolumn{2}{c}{CART} & \multicolumn{2}{c}{Weka} & \multicolumn{2}{c}{YaDT} \\ \cmidrule(r){2-9}
 &  $20\%$ & $50\%$ & $20\%$ & $50\%$ & $20\%$ & $50\%$ & $20\%$ & $50\%$ \\ \midrule
		postoperative-patient-data & $0.55$ & $0.58$ & $0.57$ & $0.56$ & $0.57$ & $0.58$ & $0.59$ & $0.63$ \\
		hayes-roth & $0.64$ & $0.71$ & $0.65$ & $0.74$ & $0.64$ & $0.74$ & $0.64$ & $0.66$ \\
		lupus & $0.64$ & $0.70$ & $0.63$ & $0.67$ & $0.66$ & $0.70$ & $0.68$ & $0.73$ \\
		appendicitis & $0.79$ & $0.82$ & $0.79$ & $0.80$ & $0.77$ & $0.79$ & $0.79$ & $0.82$ \\
		molecular\_biology\_promoters & $0.62$ & - & $0.62$ & - & $0.62$ & - & $0.62$ & - \\
		tae & $0.56$ & - & $0.57$ & - & $0.56$ & - & $0.56$ & - \\
		cloud & $0.63$ & - & $0.66$ & - & $0.64$ & - & $0.63$ & - \\
		cleveland-nominal & $0.62$ & - & $0.63$ & - & $0.64$ & - & $0.64$ & - \\
		lymphography & $0.70$ & - & $0.75$ & - & $0.73$ & - & $0.74$ & - \\
		hepatitis & $0.72$ & - & $0.76$ & - & $0.78$ & - & $0.78$ & - \\
		glass2 & $0.68$ & - & $0.70$ & - & $0.69$ & - & $0.65$ & - \\
		backache & $0.77$ & - & $0.78$ & - & $0.79$ & - & $0.81$ & - \\
		auto & $0.72$ & - & $0.74$ & - & $0.68$ & - & $0.69$ & - \\
		glass & $0.67$ & - & $0.67$ & - & $0.68$ & - & $0.68$ & - \\
		biomed & $0.79$ & $0.89$ & $0.79$ & $0.83$ & $0.78$ & $0.83$ & $0.76$ & $0.85$ \\
		new-thyroid & $0.90$ & $0.92$ & $0.90$ & $0.92$ & $0.88$ & $0.91$ & $0.88$ & $0.91$ \\
		spect & $0.78$ & - & $0.79$ & - & $0.77$ & - & $0.86$ & - \\
		breast-cancer & - & - & - & - & - & - & - & - \\
		heart-statlog & $0.71$ & - & $0.70$ & - & $0.71$ & - & $0.74$ & - \\
		haberman & - & - & - & - & - & - & - & - \\
		heart-h & - & - & - & - & - & - & - & - \\
		hungarian & $0.72$ & - & $0.75$ & - & $0.74$ & - & $0.78$ & - \\
		cleve & - & - & - & - & - & - & - & - \\
		heart-c & - & - & - & - & - & - & - & - \\
		cleveland & - & - & - & - & - & - & - & - \\
		ecoli & $0.91$ & - & $0.90$ & - & $0.91$ & - & $0.92$ & - \\
		schizo & - & - & - & - & - & - & - & - \\
		bupa & - & - & - & - & - & - & - & - \\
		colic & - & - & - & - & - & - & - & - \\
		dermatology & $0.98$ & $0.98$ & $0.98$ & $0.98$ & $0.98$ & $0.98$ & $0.98$ & $0.98$ \\
		cars & $0.87$ & - & $0.86$ & - & $0.87$ & - & $0.87$ & - \\
		soybean & $0.93$ & - & $0.91$ & - & $0.92$ & - & $0.91$ & - \\
		australian & - & - & - & - & - & - & - & - \\
		diabetes & - & - & - & - & - & - & - & - \\
		contraceptive & - & - & - & - & - & - & - & - \\
		\bottomrule
	\end{tabular}
	%}
	\caption{Overview of the percentage of correctly classified examples on the test data set. 
	Recall that for each instance the training data set consists of~$20\%$ or~$50\%$ of the examples and for both values we sampled 10~times.
	If \OS{} solved at least~$5$ of the $10$~sampled instances within the time limit, the corresponding entry is the average percentage of correctly classified examples.
	Each of the three heuristics was only used for instances which were solved by \OS{}.}
	\label{tab:overview_datasets_test_classification_comparison}
\end{table}

\begin{table}[t!]
	\centering
	%\scalebox{0.75}{
	\small
	\begin{tabular}{lrrrrrrrr}
\toprule
	\multirow{2}{*}{Instance name} & \multicolumn{2}{c}{\OS}  & \multicolumn{2}{c}{CART} & \multicolumn{2}{c}{Weka} & \multicolumn{2}{c}{YaDT} \\ \cmidrule(r){2-9}
 &  $20\%$ & $50\%$ & $20\%$ & $50\%$ & $20\%$ & $50\%$ & $20\%$ & $50\%$ \\ \midrule
		postoperative-patient-data & $0.95$ & $0.90$ & $0.92$ & $0.77$ & $0.88$ & $0.68$ & $1.00$ & $0.95$ \\
		hayes-roth & $0.94$ & $0.87$ & $0.91$ & $0.90$ & $0.88$ & $0.81$ & $0.96$ & $0.88$ \\
		lupus & $0.86$ & $0.77$ & $0.87$ & $0.80$ & $0.83$ & $0.74$ & $0.99$ & $0.99$ \\
		appendicitis & $0.97$ & $0.86$ & $0.93$ & $0.87$ & $0.94$ & $0.79$ & $1.00$ & $0.96$ \\
		molecular\_biology\_promoters & $0.98$ & - & $0.97$ & - & $0.93$ & - & $0.96$ & - \\
		tae & $0.93$ & - & $0.83$ & - & $0.75$ & - & $0.99$ & - \\
		cloud & $0.92$ & - & $0.92$ & - & $0.87$ & - & $0.94$ & - \\
		cleveland-nominal & $0.91$ & - & $0.89$ & - & $0.78$ & - & $0.87$ & - \\
		lymphography & $0.96$ & - & $0.96$ & - & $0.91$ & - & $0.93$ & - \\
		hepatitis & $0.93$ & - & $0.90$ & - & $0.88$ & - & $0.95$ & - \\
		glass2 & $0.92$ & - & $0.92$ & - & $0.80$ & - & $0.86$ & - \\
		backache & $0.92$ & - & $0.91$ & - & $0.92$ & - & $0.94$ & - \\
		auto & $0.91$ & - & $0.89$ & - & $0.74$ & - & $0.79$ & - \\
		glass & $0.93$ & - & $0.89$ & - & $0.87$ & - & $0.91$ & - \\
		biomed & $0.90$ & $0.84$ & $0.90$ & $0.78$ & $0.84$ & $0.71$ & $0.93$ & $0.86$ \\
		new-thyroid & $0.94$ & $0.90$ & $0.93$ & $0.86$ & $0.87$ & $0.79$ & $0.97$ & $0.82$ \\
		spect & $0.83$ & - & $0.81$ & - & $0.71$ & - & $0.95$ & - \\
		breast-cancer & - & - & - & - & - & - & - & - \\
		heart-statlog & $0.85$ & - & $0.90$ & - & $0.80$ & - & $0.85$ & - \\
		haberman & - & - & - & - & - & - & - & - \\
		heart-h & - & - & - & - & - & - & - & - \\
		hungarian & $0.87$ & - & $0.90$ & - & $0.85$ & - & $0.96$ & - \\
		cleve & - & - & - & - & - & - & - & - \\
		heart-c & - & - & - & - & - & - & - & - \\
		cleveland & - & - & - & - & - & - & - & - \\
		ecoli & $0.92$ & - & $0.94$ & - & $0.88$ & - & $0.94$ & - \\
		schizo & - & - & - & - & - & - & - & - \\
		bupa & - & - & - & - & - & - & - & - \\
		colic & - & - & - & - & - & - & - & - \\
		dermatology & $1.00$ & $0.92$ & $1.00$ & $0.99$ & $0.94$ & $0.97$ & $1.00$ & $1.00$ \\
		cars & $0.84$ & - & $0.87$ & - & $0.75$ & - & $0.80$ & - \\
		soybean & $0.80$ & - & $0.87$ & - & $0.86$ & - & $0.90$ & - \\
		australian & - & - & - & - & - & - & - & - \\
		diabetes & - & - & - & - & - & - & - & - \\
		contraceptive & - & - & - & - & - & - & - & - \\
		\bottomrule
	\end{tabular}
	%}
	\caption{Overview of the balance of the trees of the solutions. 
	Recall that for each instance the training data set consists of~$20\%$ or~$50\%$ of the examples and for both values we sampled 10~times.
	If \OS{} solved at least~$5$ of the $10$~sampled instances within the time limit, the corresponding entry is the average balance.
	Each of the three heuristics was only used for instances which were solved by \OS{}.}
	\label{tab:overview_datasets_tree_balance_comparison}
\end{table}

\subsection{Solution quality}

Although it is not our main goal in this work, we want to evaluate the solution quality of the minimum-size decision trees that \OS{} can compute to see whether there are any potential tendencies as compared to heuristics.
Note that, in the following, we will focus only on the instances where the minimum-size trees could be computed by \OS{}.
\OS{} can be used in an anytime fashion to successively improve the tree size, potentially yielding small-yet-suboptimal tree sizes with reasonable running time, but we leave exploring this direction to future work.

We compare the obtained minimum-size trees with the results of three popular heuristics.
More precisely, we compare against the CART heuristic~\citep{DBLP:books/wa/BreimanFOS84} and the C4.5 heuristic~\citep{Quinlan93}.
For C4.5 we compare against 2 variants: the one in Weka~\citep{WittenFH16} and the one in YaDT~\citep{Ruggieri04,Ruggieri19}, from now on simply called Weka and YaDT.
We use two variants of C4.5 since Weka can be parameterized to compute a perfect decision tree while YaDT cannot.
In other words, CART and Weka also compute perfect decision trees while YaDT does not.

\Cref{tab:overview_datasets_tree_size_comparison} shows the average solution size for all four algorithms.
The decision trees computed by our algorithm \OS{} are usually smaller than the ones computed by CART and Weka. 
On average our trees have roughly one vertex less.
In contrast, the imperfect trees computed by YaDT are still smaller than the ones computed by \OS; on average they have roughly two vertices less.

\Cref{tab:overview_datasets_test_classification_comparison} shows the percentage of correctly classified examples on the test data set, which consists of all examples which are not contained in the training data set.
All algorithms perform very similar; YaDT achieves an on average 1.03\,\% lower ratio of misclassifications as compared to minimum-size trees.

\Cref{tab:overview_datasets_tree_balance_comparison} shows the average balance of the computed trees.
For this, we use the \emph{normalized inverse Sackin index}~\citep{LLMN22} of a binary rooted tree~$T$:
To define this value, we also need the \emph{Sackin index}~$\sack(T)$ which is defined as~$\sack(T)=\sum_{\ell\in L(T)}\nu(\ell)$.
Here,~$L(T)$ is the set of leafs of~$T$ and~$\nu(\ell)$ is the number of internal vertices from leaf~$\ell$ to the root of~$T$.
Now, let~$T'$ be a balanced binary tree with the same number of vertices as~$T$ (note that~$T'$ is not necessarily complete).
The \emph{normalized inverse Sackin index}~$\nISack(T)$ of a binary rooted tree~$T$ is~$\nISack(T)=\sack(T)/(|L(T)|\cdot \sack(T'))$.
Thus,~$0\le \nISack(T)\le 1$ for each~$T$ and~$\nISack(T)=1$ if and only if~$T$ is balanced.

\Cref{tab:overview_datasets_tree_balance_comparison} shows that all computed trees have a large normalized inverse Sackin index, that is, all of them are very balanced.
Usually, \OS{} achieves a higher score than CART and Weka; there are only a few instances where the trees computed with these heuristics are more balanced.
In contrast, the normalized inverse Sackin index of the trees computed by YaDT is usually even larger than that of the trees computed by \OS.

\begin{figure}
\centering
\begin{minipage}{.45\textwidth}
	\centering
	\includegraphics[width=\linewidth]{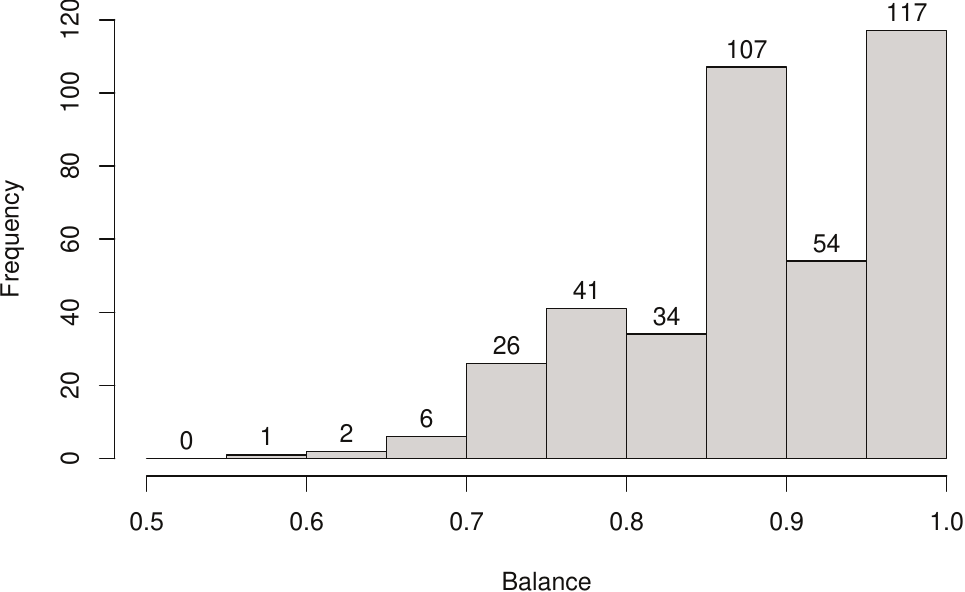}
	\caption{Balance of the size-perfect decision trees found by \OS{}.\\}
	\label{fig-balancing}
\end{minipage}%
\hspace{0.3cm}
\begin{minipage}{.45\textwidth}
	\centering
	\includegraphics[width=\linewidth]{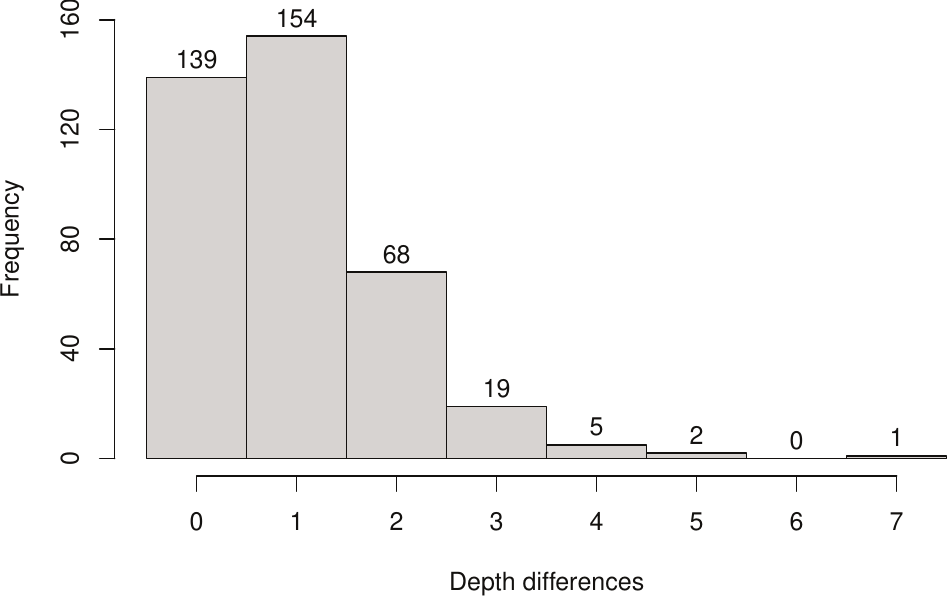}
	\caption{Differences between the depth of the perfect decision trees found by \OS{} and the smallest depth of any perfect decision tree.}
	\label{fig-depth-diffs}
\end{minipage}
\end{figure}

\Cref{fig-balancing} shows a histogram on the normalized inverse Sackin index of the solutions found by \OS.
The normalized inverse Sackin index of all trees is larger than~$0.5$ and more than $85\%$~of the solutions have a normalized inverse Sackin index of at least~$0.8$.
Furthermore, more than $20\%$~ of the solutions are almost complete binary trees.
Hence, the solutions found by \OS{} are usually very tree-like (have a high normalized inverse Sackin index value) and only in a few cases the optimal solution has low normalized inverse Sackin index.

This is also supported by \Cref{fig-depth-diffs} which shows a histogram for the differences between the depth of the trees found by \OS{} and the smallest depth of any perfect decision tree. To calculate the optimal depth we used \texttt{MurTree}. We can see that roughly~$75\%$ of the trees found by \OS{} are at most one away from the optimal depth. While a close to optimal depth does not necessarily mean that the tree is balanced, it does suggest that the tree is probably almost as balanced as it can possibly be while still being a perfect decision tree.

Overall, our evaluation indicates that, while being smaller than the perfect trees computed by CART and Weka, there is no clear tendency in terms of solution quality between \OS{}, CART, and Weka.
The trees computed by YaDT are even smaller than the minimum-size perfect trees, there is no clear tendency in terms of classification error on the test data set, but they achieve better balance scores.

From these results we can derive two hypotheses to test in future work:
First, whenever one is interested in perfect decision trees, minimum-size trees do not compromise solution quality as compared to the CART and Weka heuristics.
Second, minimum-size perfect decision trees achieve similar misclassification errors in validation sets as compared to the non-perfect decision trees computed by the YaDT heuristic.
Validating in particular the second hypothesis would show that fast heuristics such as YaDT are highly suitable for practice.

}

\section{Outlook}
\label{sec:conclusion}

We have provided a new fastest solver \OS{} for computing perfect decision trees with a size constraint.
While previous algorithms~\citep{DBLP:conf/ijcai/NarodytskaIPM18,DBLP:conf/sat/JanotaM20,DBLP:journals/jmlr/DemirovicLHCBLR22} only support binary dimensions, \OS{} in particular benefits from dimensions having many thresholds.
We conclude with a set of limitations and possible extensions of \OS{}.

First, we focused on binary classification problems.
However, \OS{} can easily be adapted for more classes: any example having a different class than that of the witness in any leaf is a dirty example; it would be interesting to scrutinize the resulting algorithm on multiclass instances.

Second, it is interesting to adapt and tune \OS{} to find decision trees with a depth constraint. %~\cite{DBLP:conf/cp/BessiereHO09,NijssenF10,DBLP:conf/aaai/VerwerZ19,carrizosa_mathematical_2021,AglinNS20,avellaneda_efficient_2020,DBLP:journals/jmlr/DemirovicLHCBLR22,SCM23,DBLP:journals/air/CostaP23}.
While this constraint can easily be incorporated into \OS{} and all of our improvements are still valid, they could be made much tighter and many new improvements are possible that do not apply to the size constraint.
For example, once a leaf with maximal depth has been found, the corresponding leaf-to-root path $P$ cannot change anymore.
The solution subtrees rooted at vertices in $P$ thus have fixed example sets and can be determined independently of each other.

%\looseness=-1
Finally, it is essential to adapt \OS{} to looser accuracy guarantees, for instance, to the scenario where a given number~$\tau$ of misclassifications are allowed.
\citet{GZ24} showed that the underlying training problem is tractable in theory but the running time of their algorithm is impractical.
% Furthermore, note that \citet{DBLP:journals/jmlr/DemirovicLHCBLR22} provided an implementation for the related problem of minimizing the number of having for a given depth constraint.
The witness-tree paradigm could be extended to this problem by replacing dirty examples with a set of $\tau + 1$ dirty examples, of which one needs to be reclassified.
While the theoretical running-time guarantee would increase, in practice this drawback could be outweighed by the fact that the size $s$ of the optimal tree could decrease substantially (as shown by the results of YaDT on the benchmark data set).
% Adapting \OS{} seems challenging since 1) most of our lower bounds heavily exploit the fact that each example needs to be classified correctly, and 2) it is not sufficient anymore to pick a single dirty example~$e$ since~$e$ might be an example which is misclassified; instead one has to choose a set of at least $\tau+1$~dirty examples to have a guarantee that at least one of them needs to be classified correctly.

\section*{Acknowledgments}
Luca Pascal Staus was supported by the Carl Zeiss Foundation, Germany, within the project ``Interactive Inference''.
Frank Sommer was supported by the Alexander von Humboldt Foundation and partially supported by the  DFG, project EAGR (KO 3669/6-1).

%\ifarxiv
\newpage
\appendix
\onecolumn
\section*{Appendix}
\appendixProofText
%\else
%\fi

\end{document}